\documentclass[draftclsnofoot]{IEEEtran}
 \onecolumn
\usepackage{latexsym}
\usepackage{cite}
\usepackage{color}
\usepackage{bm}
\usepackage{amsmath, amssymb}
\usepackage[dvips]{graphics}
\usepackage{graphicx}
\usepackage{psfrag}
 \allowdisplaybreaks

\newtheorem{definition}{Definition}

\newtheorem{theorem}{Theorem}
\newtheorem{lemma}[theorem]{Lemma}

\newtheorem{corollary}[theorem]{Corollary}

\begin{document}
\title{MAC with Action-Dependent State Information at One Encoder}
\date{August, 2012}
\author{Lior Dikstein, Haim H. Permuter and Shlomo (Shitz) Shamai}
\maketitle \footnotetext[1]{Parts of this paper appeared in the IEEE International Symposium on
Information Theory (ISIT 2012),Cambridge, MA, US, July 2012 and at the IEEE 27th Convention of
Electrical and Electronics
Engineers in Israel (IEEEI 2012), Nov. 2012.}%
 \footnotetext[2]{This work has been supported by the CORNET Consortium Israel Ministry for
Industry and Commerce.}

\begin{abstract}
Problems dealing with the ability to take an action that affects the states of state-dependent
communication channels are of timely interest and importance. Therefore, we extend the study of
action-dependent channels, which until now focused on point-to-point models, to multiple-access
channels (MAC). In this paper, we consider a two-user, state-dependent MAC, in which one of the
encoders, called the informed encoder, is allowed to take an action that affects the formation of
the channel states. Two independent messages are to be sent through the channel: a common message
known to both encoders and a private message known only to the informed encoder. In addition, the
informed encoder has access to the sequence of channel states in a non-causal manner. Our
framework generalizes previously evaluated settings of state dependent point-to-point channels
with actions and MACs with common messages. We derive a single letter characterization of the
capacity region for this setting. Using this general result, we obtain and compute the capacity
region for the Gaussian action-dependent MAC. The unique methods used in solving the Gaussian
case are then applied to obtain the capacity of the Gaussian action-dependent point-to-point
channel; a problem was left open until this work. Finally, we establish some dualities between
action-dependent channel coding and source coding problems. Specifically, we obtain a duality
between the considered MAC setting and the rate distortion model known as ``Successive Refinement
with Actions''. This is done by developing a set of simple duality principles that enable us to
successfully evaluate the outcome of one problem given the other.

\end{abstract}


\section{Introduction}
STATE-DEPENDENT channels model a communication situation where the channel is time variant. Such channels characterize a significant collection of communication scenarios, ranging from interfering transmissions models to cases where the states are generated by nature. Problems of coding for these channels have received much attention due to the wide range of their potential applications. These applications vary from modeling communication links such as fading, to interference in a wireless network. Furthermore, problems relating Multiple Access Channels (MAC) with channel state information (CSI) have been thoroughly studied due to their importance in modeling wireless communication systems.  Most of the channels investigated until now have been examined under the assumption that the states affecting the channel cannot be influenced by the communication system.

\par In this paper, we consider an action-dependent channel. The motivation for studying these channels stems from the implications and practical use of the 'action' in modeling important communication scenarios. For instance, one interpretation of the action could be a noisy public relay. In this case, the relay outputs are modeled as a function of the messages. We provide a relay output sequence to be transmitted: $A^n(M_1,M_2)$ and obtain the state sequence, $S^n$, via the memoryless noisy transformation $p(s|a)$. The relay outputs are public and, therefore, monitored beforehand so that $S^n$ is known at the transmitter. Our model is the natural extension of that for the two-user cognitive setting.

So far the most studies on action-dependent channels focused on point-to-point channels. In this
work, we broaden the research of action-dependent channels to MACs. We consider a MAC where one
of the encoders is allowed to take an action that affects the formation of the channel states.
\begin{figure}[h!]
    \begin{center}
        \begin{psfrags}
            \psfragscanon
            \psfrag{A}[][][0.8]{\ \ \ \ \ \ \ \ \ \ \ \ \ \ \ \ \ \  Uninformed}
            \psfrag{B}[][][0.8]{\ \ \ \ \ \ \ \ \ \ \ \ \ \ \ \ \ \ Informed}
            \psfrag{C}[][][0.8]{\ \ \ \ \ \ \ \ \ \ \ \ \ \ \ \ \ \ MAC}
            \psfrag{D}[][][0.8]{\ \ \ \ \ \ \ \ \ \ \ \ \ \ \ \ \ \ \ \ \ $p(y|x_1,x_2,s)$}
            \psfrag{E}[][][0.8]{\ \ \ \ \ \ \ \ \ \ \ \ \ \ Decoder}
            \psfrag{F}[][][0.8]{\ \ \ \ \ \ \footnotesize{$X_1^n(M_1)$}}
            \psfrag{G}[][][0.8]{\ \ \ \ \ \ \ \ \ \footnotesize{$X_2^n(M_1,M_2)$}}
            \psfrag{H}[][][0.8]{$Y^n$}
            \psfrag{I}[][][0.8]{\ \ \ \ \ \ \ $p(s|a)$}
            \psfrag{J}[][][0.8]{$M_1$}
            \psfrag{K}[][][0.8]{$M_2$}
            \psfrag{L}[][][0.8]{\ \ \ $(\hat{M}_1,\hat{M}_2)$}
            \psfrag{M}[][][0.8]{\ \ \ \ \ \ \ \ \ \ \ \ \ \ \ \ \ \ Encoder}
            \psfrag{N}[][][0.8]{\ \ \ \ \ \ \ \ \ \ \ \ \ \ \ \ \ \ Encoder}
            \psfrag{O}[][][0.8]{$A^n(M_1,M_2)$}
            \psfrag{P}[][][0.8]{$S^n$}
            \centerline{\includegraphics[scale = .6]{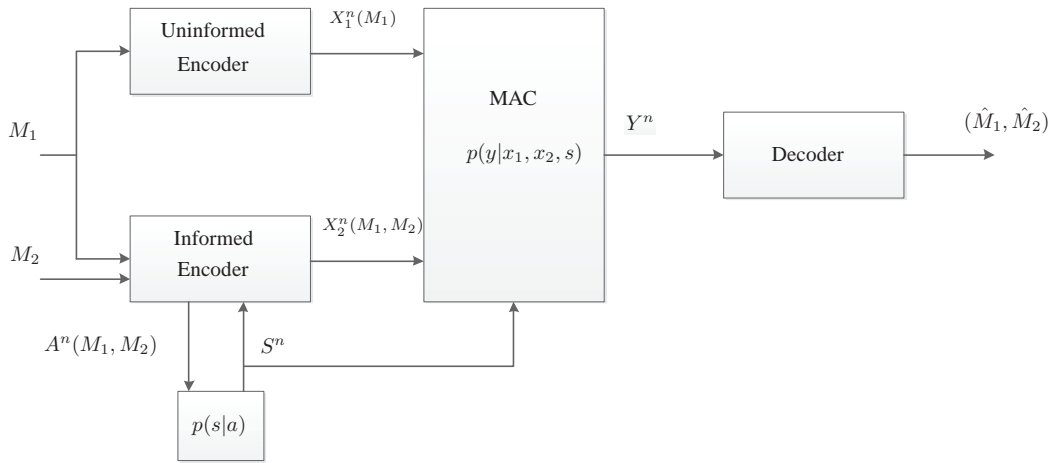}}
            \caption{The MAC with action-dependent state information at one encoder, which is considered in this paper.} \label{channel1}
            \psfragscanoff
        \end{psfrags}
     \end{center}
 \end{figure}
Specifically, we examine a MAC communication system, that is illustrated in Fig. \ref{channel1},
where two encoders have access to a common message and only one, the informed encoder, has access
to a private message. This encoder can generate an action sequence dependent on both messages
that, in turn, affects the channel states. Furthermore, the states affected by the action
sequence are accessible non-causally to the informed encoder when producing the channel input. We
will refer to this model as the Action-MAC. We characterize the capacity region of this channel
for the general finite alphabet case with a single letter expression. This is done by using a
random coding scheme and a binning technique, where the encoding of the messages is done in three
parts. In the first part, the uninformed encoder encodes a message using its signal, $X_1$. In
the second part, the informed encoder encodes a message using the actions where $X_1$ is used as
side information. In the third stage, a Gel'fand-Pinsker coding is done from the informed encoder
where the action and $X_1$ are used as side information.


\par The study of state-dependent channels dates back to Shannon \cite{Shannon}, who first introduced and characterized the capacity of a state-dependent, memoryless point-to-point channel with independent and identically distributed (i.i.d.) states available causally at the encoder. Gel'fand and Pinsker \cite{Gelfand-Pinsker}, and later Heegard and El Gamal \cite{Heegard-Gamal}, studied a case in which the encoder observes the channel states non-causally. They derived a single letter formula for the capacity, using a binning coding scheme. The main idea of this scheme is to generate a subcodebook for each message. Next, in order to send a message we send one of the codewords in the subcodebook that is jointly typical with the sequence of channel states. The results obtained by Gel'fand and Pinsker were used by Costa in his famous ``Writing on Dirty Paper'' \cite{Costa}. He applied these results to the case where there are two additive Gaussian noise sources, where one of the noises represents the interferences and is modeled by the channel's states.

\par An extension of this model, known as the Generalized Gel'fand-Pinsker (GGP) MAC, was studied by Somekh-Baruch, Shamai and Verdu, in their paper \cite{Baruch-Shamai-Verd}. They considered a MAC with common and private messages where the encoder informed of the private message is additionally informed of the channel states non-casually. They characterized the capacity region for the general finite alphabet case using a generalized binning coding scheme. A MAC with private messages at both encoders and state information known to one encoder channel was examined in \cite{Kotagiri07multipleaccess}, where an inner bound for the capacity region in the general discrete memoryless case was found. Further MAC models, where the states are known causally or strictly causally were considered by Lapidoth and Steinberg in \cite{DBLP:journals/corr/abs-1106-0380},\cite{5513459} and later by Li, Simeone and Yener in \cite{DBLP:journals/corr/abs-1011-6639}. Another version of a MAC, where the states are known noncausally at one encoder and causally at the second encoder, was considered by Zaidi, Piantanida and Shamai in \cite{DBLP:journals/corr/abs-1105-5975}.

\par The novel idea of an action-dependent state scenario was introduced in the work of Wiessman \cite{Tsachy-Weissman}. In his paper, he considered a point-to-point channel where the encoder is allowed to take an action that can affect the channel's states. He characterized the capacity for the case where the channel inputs are allowed to depend causally or non-causally on the state sequence. Furthermore, the Gaussian case for this channel was also introduced, however, only an achievable scheme was given and the capacity for this case remained unsolved.

\par In this work, we consider a MAC with a common and private message along with CSI at one encoder. However in our setting, this encoder can now take an action that affects the states. This  generalizes the results presented in the works of Somekh-Baruch, Shamai and Verdu \cite{Baruch-Shamai-Verd} as well as Wiessman \cite{Tsachy-Weissman}. We particularly focus on analyzing the Gaussian case for our channel model. We find the capacity region and compute it. In the process, the new results obtained help us find the capacity expression for the Gaussian point-to-point action-dependent channel left open in \cite{Tsachy-Weissman}. A similar result was also derived by obtaining a correspondence between the action-dependent point-to-point channel to the GGP MAC with only a common message. This same correspondence was also obtained simultaneously and independently by Choudhuri and Mitra in \cite{Choudhuri-Mitra}. Furthermore, we investigate the dual relationships between channel coding and source coding problems with actions. We establish a duality between our Action-MAC to the rate distortion model known as ``Successive Refinement with Actions'', presented in \cite{Khiang-Asnani-Weissman} and explore the similarities between them. We show how using the results gained from the solution of the Action-MAC assists us when analyzing the ``Successive Refinement with Actions'' setting and vice versa. We obtain a set of duality principles which help predict the outcome of one problem given the other. These set of principle also help establish more dualities between known channel coding and source coding problems with actions.

\par The remainder of the paper is organized as follows. Section \ref{Notation} presents the mathematical notation used in this paper and the exact formulation of the Action-MAC setting. In Section \ref{Main}, we state our capacity results, which include the capacity region of the Action-MAC, and discuss additional special cases that show consistency with previous works. Sections \ref{Achievability} and \ref{Converse} are devoted to describing the achievability coding scheme and the converse proof respectively. Section \ref{GaussMac} examines the Gaussian case for the Action MAC, where we find the capacity region for this specific setting and compute the region for different parameters. Here, we also find the capacity expression for the point-to-point model. In Section \ref{DualSection}, rate distortion coding duals are considered. We conclude in Section \ref{Conclusion} with a summary of this work.


\section{Notation and Problem Definition}\label{Notation}

Throughout the paper, random variables will be denoted by upper case letters, deterministic realizations or specific values will be denoted by lower case letters and calligraphic letters will denote the alphabets of the random variables. Let $x^n$ denote vectors of $n$ elements, i.e. $x^n=(x_1,x_2,...,x_n)$ and $x_i^j$ denote the $i-j+1$-tuple $(x_i,x_{i+1},...,x_j)$ when $j\geq i$ and an empty set otherwise. The probability distribution function of $X$, the joint distribution function of $X$ and $Y$ and the conditional distribution of $X$ given $Y$ will be denoted by $P_{X}$, $P_{X,Y}$ and $P_{X|Y}$, respectively.

\par We consider a channel coding MAC with action-dependent states, where the states are known non-causally at one encoder, as illustrated in Fig. \ref{channel1}. The Action-MAC setting consists of two transmitters (encoders) and one receiver (decoder). Let $n$ denote the block length and $\mathcal{A}, \mathcal{S}, \mathcal{X}_1, \mathcal{X}_2$ and $\mathcal{Y }$ be finite sets which denote the actions, states, the uninformed encoder's inputs, the informed encoder's inputs, and the outputs, respectively. A state information channel is described by a triple $(\mathcal{A},P_{S|A},\mathcal{S})$ and is assumed to be memoryless with transition probabilities:
\begin{equation}
p(s_i|a^i,s^{i-1},m_1,m_2)=p(s_i|a_i).
\end{equation}
\begin{definition}
A $((2^{nR_1},2^{nR_2}),n)$ code for the channel in Fig. \ref{channel1} consists of two sets of integers, $\mathcal{M}_1 = \{1,2,...,2^{nR_1}\}$ and $\mathcal{M}_2 = \{1,2,...,2^{nR_2}\}$, called message sets. An index is chosen uniformly and independently by the senders out of the message sets. The uninformed encoder selects a channel input sequence, $X_1^n = X_1^n(M_1)$. Given the messages $M_1,M_2$, an action sequence denoted  $A^n = A^n(M_1,M_2)$ is selected by the informed encoder. A state sequence, $S^n$, is then selected by the channel, with $A^n$ being the input chosen by the informed encoder. Next, a channel input sequence, $X_2^n = X_2^n(M_1,M_2,S^n)$, is selected. The output of the channel is denoted $Y^n$. The channel is characterized by the conditional probability $p(y_i|x_{1,i},x_{2,i},s_i)$ and is assumed to be memoryless. Therefore, both probabilities do not depend on the index $i$, i.e.
\begin{equation}
p(y_i,s_{i+1}|x_1^i,x_2^i,s^i)=p(y_i|x_{1,i},x_{2,i},s_i)p(s_{i+1}).
\end{equation}
The code is defined by the encoding functions:
\begin{equation}
f_1:\mathcal{M}_1\rightarrow \mathcal{X}_1^n\\
\end{equation}
and
\begin{equation}
f_2:\mathcal{M}_1\times\mathcal{M}_2\times\mathcal{S}^n\rightarrow \mathcal{X}_2^n,
\end{equation}
an action encoder
\begin{equation}
f_A:\mathcal{M}_1\times\mathcal{M}_2\rightarrow \mathcal{A}^n
\end{equation}
and a decoding function
\begin{equation}
g:\mathcal{Y}^n\rightarrow (\hat{\mathcal{M}_1}\times\hat{\mathcal{M}_2}).
\end{equation}
\end{definition}
We define the average probability of error for the $((2^{nR_1},2^{nR_2}),n)$ code as follows:
\begin{equation}
P_e^{(n)} = \frac{1}{2^{n(R_1+R_2)}}\sum_{m_1,m_2 \in \mathcal{M}_1 \times \mathcal{M}_2} \Pr\{g(Y^n)\neq (m_1,m_2)|(m_1,m_2) \text{ sent}\}.
\end{equation}
A pair rate $(R_1,R_2)$ is achievable if there exists a sequence of codes $(2^{nR_1},2^{nR_2},n)$ s.t. $P^{(n)}_e \rightarrow 0$.\\
The capacity region is the closure of all achievable rates.


\section{Capacity Results for the Action-MAC}\label{Main}
\subsection{Capacity Region}
The following theorem provides an expression for the capacity region of the MAC with action-dependent state information at one encoder channel, for finite alphabets $\mathcal{A},\mathcal{S},\mathcal{X}_1,\mathcal{X}_2$:
\begin{theorem} \label{TheoremMain}
The capacity region of the MAC with action-dependent state information at one encoder, as shown in Fig. \ref{channel1}, is the closure of the set that contains all the rates that satisfy
\begin{eqnarray}
R_2&\leq& I(U;Y|X_1)-I(U;S|A,X_1)\nonumber\\
R_1+R_2&\leq& I(X_1,U;Y)-I(X_1,U;S|A), \label{CapacityRegion}
\end{eqnarray}
for some joint probability distribution of the form
\begin{equation}
P_{A,S,U,X_1,X_2,Y} = P_{X_1}P_{A|X_1}P_{S|A}P_{U|S,A,X_1}P_{X_2|X_1,S,U}P_{Y|X_1,X_2,S}\label{distribution}
\end{equation}
and  $|U|\leq |\mathcal{A}||\mathcal{S}||\mathcal{X}_1||\mathcal{X}_2|+1$.
\end{theorem}

\begin{lemma}\label{lemma1}
The capacity region described in {\it Theorem 1}, given in (\ref{CapacityRegion}), is convex.
\end{lemma}
The proof of Lemma \ref{lemma1} is given in appendix \ref{LemmaProof}.

In the following corollaries we present an alternative representation for the capacity region given in Theorem \ref{TheoremMain}, which provides a more intuitive view. Furthermore, we derive the corner points for the Action-MAC rate region and the distribution of rate resources.
\begin{corollary}\label{intuitiveRegion}
The following region is equivalent to the one presented in {\it Theorem 1}:
\begin{eqnarray}
R_2&\leq& I(A,U;Y|X_1)-I(U;S|X_1,A)\nonumber\\
R_1+R_2&\leq& I(X_1,A,U;Y)-I(X_1,U;S|A)\label{Corollary1}.
\end{eqnarray}
\begin{proof}
It is clear that the region (\ref{Corollary1}) contains the one in (\ref{CapacityRegion}) since $I(U;Y|X_1)\leq I(A,U;Y|X_1)$ and $I(X_1,U;Y)\leq I(X_1,A,U;Y)$. For the converse, we have:
\begin{align}
I(A,U;Y|X_1)-I(U;S|X_1,A)&=I(A,U;Y|X_1)-I(A,U;S|X_1,A)\nonumber\\
                         &=I(\tilde{U};Y|X_1)-I(\tilde{U};S|X_1,A),
\end{align}
taking $\tilde{U}=(A,U)$, where $(X_1,S,A,\tilde{U},X_2,Y)$ satisfies the same joint distribution relations as $(X_1,S,A,U,X_2,Y)$ .
\end{proof}
\end{corollary}
Furthermore, notice that if we rearrange the expressions in (\ref{Corollary1}) by using the chain rule and the Markov $X_1-A-S$, we can express the capacity region as:
\begin{eqnarray}
R_2 &\leq& I(A;Y)-I(X_1;Y) + I(X_1,U;Y|A) - I(X_1,U;S|A)\nonumber\\
R_1 + R_2 &\leq& I(A;Y) + I(X_1,U;Y|A) - I(X_1,U;S|A).\label{Intuition}
\end{eqnarray}

This result provides intuition, and makes immediate sense. First, the informed encoder transmits information using the action sequence $A$ at rate $I(A;Y)$. This first transmission is decoded and used at the decoder to decode a second transmission; hence the conditioning. Next, by Gel'fand-Pinsker given $A$: $X_1$ and $U$ can be decoded. Namely, we notice the same expression, $I(A;Y) + I(X_1,U;Y|A) - I(X_1,U;S|A)$, appears in the bound on $R_2$ as well as the bound on $R_1+R_2$. However, in the expression bounding $R_2$ we omit the rate $I(X_1;Y)$, which is the rate used by the uninformed user.

\begin{corollary}\label{corollary2}
Another presentation of the capacity region in Theorem \ref{TheoremMain} is:
\begin{eqnarray}
R_2 &\leq& I(A;Y|X_1) + I(Y;U|A,X_1) - I(S;U|A,X_1)\nonumber\\
R_1 + R_2 &\leq& I(X_1;Y)+I(A;Y|X_1)+ I(Y;U|A,X_1) - I(S;U|A,X_1),\label{DualCap}
\end{eqnarray}
which follows from a different reorganization of the region presented in Corollary \ref{intuitiveRegion} by using the chain rule together with the Markov  property $X_1-A-S$.
\end{corollary}
\par This is another intuitive perception for our rate region. Here, we can recognize that the uninformed user sends information at rate $I(X_1;Y)$. This transmission is then decoded at the decoder and used as side information. Next, the informed user sends information about the action at a rate $I(A;Y|X_1)$ (since $X_1$ was decoded and is now known at the decoder - hence the conditioning). Finally, now that $(X_1,A)$ are known at the decoder and hence conditioned upon, the rest of the information is sent in a Gel'fand Pinsker-like scheme at rate $I(U;Y|X_1,A) - I(U;S|X_1,A)$. Namely, we can  look at the conditioned expressions of the region as standard expressions of a ``MAC with common message and state information known to one encoder'' channel, presented in \cite{Baruch-Shamai-Verd}, i.e.
\begin{eqnarray}
R_2&\leq& I(U;Y|X_1)-I(U;S|X_1)\nonumber\\
R_1+R_2&\leq&I(X_1,U;Y)-I(X_1,U;S)\nonumber
\end{eqnarray}
with a supplement of side information $A$.
\par This presentation also helps us to recognize the corner points of our region. The region is illustrated in Fig. \ref{CornerFig}. The corner points $(R_1,R_2)$ are presented in Table \ref{CornerPointsTable}.

\begin{table}[h!]
\caption{Corner points $(R_1,R_2)$ of the Action-MAC setting}\label{CornerPointsTable}
\hspace{15mm}\begin{tabular}{|c|c|}
\hline
\textbf{$R_1$} & \textbf{$R_2$}\\ \hline\hline
 $I(X_1;Y)+I(A;Y|X_1)+ I(Y;U|A,X_1) - I(S;U|A,X_1)$ & $0$\\
$I(X_1;Y)$ & $I(A;Y|X_1) + I(U;Y|X_1,A) - I(U;S|X_1,A)$\\
\hline
\end{tabular}
\end{table}

\begin{figure}[h!]
\begin{center}
\begin{psfrags}
    \psfragscanon
    \psfrag{b}{$R_1$}
    \psfrag{a}{$R_2$}
    \psfrag{f}{$I(X_1;Y)$}
    \psfrag{h}{$I(A;Y) + I(X_1,U;Y|A) - I(X_1,U;S|A)$}
    \psfrag{d}{$ I(A;Y|X_1) + I(U;Y|X_1,A) - I(U;S|X_1,A)$}
\includegraphics[scale = 0.8]{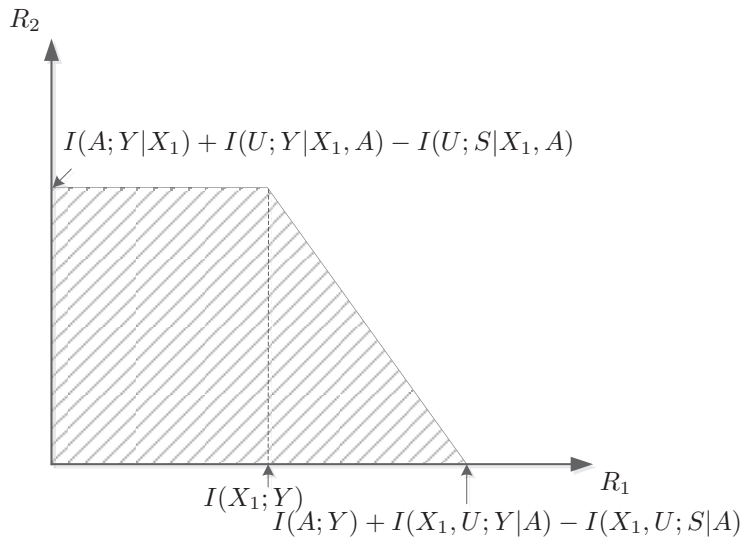}
\caption{The rate region for the Action-MAC} \label{CornerFig}
\psfragscanoff
\end{psfrags}
\end{center}
\end{figure}


\subsection{Special cases}
Before proving the theorem, let us examine some special cases in order to gain some insight and persuade ourselves that the following region is consistent with previous results.
\par {\it Case 1: The point-to-point channel with action-dependent states.} For this case, we take $R_1=0$, $P_{A|X_1}=P_A$, $P_{U|S,A,X_1}=P_{U|S,A}$, $P_{X_2|X_1,S,U}=P_{X_2|S,U}$. Hence, the region in Theorem \ref{TheoremMain} becomes:
\begin{eqnarray}
R_2&\leq& I(U;Y)-I(U;S|A)\nonumber
\end{eqnarray}
for a probability distribution of the form $P_AP_{S|A}P_{U|S,A}P_{X_2|S,U}P_{Y|X_2,S}$.
This result is the action-dependent point-to-point channel capacity discussed in  \cite{Tsachy-Weissman}.

\par {\it Case 2: MAC with a common message and a private message known to one encoder}, i.e. $|\mathcal{S}|=1$. For this case, we take $U=X_2$ and we have $I(U;S|A)=0$ and $I(X_1,U;S|A)=0$. Furthermore, there is no point in allocating resources to the action sequence. As a result, we obtain the following region:
\begin{eqnarray}
R_2&\leq&I(X_2;Y|X_1)\nonumber\\
R_1+R_2&\leq&I(X_1,X_2;Y),
\end{eqnarray}
for some joint probability distribution of the form $P_{X_1}P_{X_2|X_1}P_{Y|X_1,X_2}$.
\par {\it Case 3: Cooperative multiple-access encoding with states available at one transmitter} \cite{Baruch-Shamai-Verd}. Assume that we have a malfunction at the action encoder, i.e. we cannot choose an action that affects the formation of the states, but the Encoder 2 still knows the states noncausally. For this case $A = \emptyset$; therefore, the following expressions, $I(U;S|A)$ and $I(X_1,U;S|A)$, become $I(U;S)$ and $I(X_1,U;S)$, respectively. Hence, we have the following capacity region
\begin{eqnarray}
R_2&\leq& I(U;Y|X_1)-I(U;S|X_1)\nonumber\\
R_1+R_2&\leq&I(X_1,U;Y)-I(X_1,U;S),
\end{eqnarray}
where the probability distribution $P_{X_1}P_{A|X_1}P_{S|A}P_{U|S,A,X_1}P_{X_2|X_1,S,U}P_{Y|X_1,X_2,S}$ reduces to $P_{X_1}P_{S}P_{U|S,X_1}P_{X_2|X_1,S,U}P_{Y|X_1,X_2,S}$.
This is, indeed, the result of the Generalized Gel'fand Pinsker channel, introduced in \cite{Baruch-Shamai-Verd}.


\section{Proof of Achievability}\label{Achievability}
In this section we prove the achievability part of Theorem \ref{TheoremMain}. Throughout the achievability proof we use the definition of a strong typical set \cite{gamal2011network}. The set $\mathcal{T}_\epsilon^{(n)}(X,Y,Z)$ of $\epsilon$-typical $n$-sequences is defined by $\{(x^n,y^n,z^n):\frac{1}{n}|N(x,y,z|x^n,y^n,z^n)-p(x,y,z)|\leq \epsilon\cdot p(x,y,z)$ $\forall(x,y,z)\in \mathcal{X}\times\mathcal{Y}\times\mathcal{Z}\}$, where $N(x,y,z|x^n,y^n,z^n)$ is the number of appearances of $(x,y,z)$ in the $n$-sequence $(x^n,y^n,z^n)$.
\par The main idea of the proof is based on a random coding scheme, where the encoding of the messages is done in three parts. First, the uninformed encoder transmits $X_1^n(m_1)$ from it's code-book at rate $I(X_1;Y)$. Secondly, the informed encoder chooses an action sequence $A^n$. As a result, a state $S^n$ is generated. The action sequence is then sent at rate $I(A;Y|X_1)$, where $X_1$ is treated as side information known at both the informed encoder and the decoder. Next, the informed encoder can transmit at rate $I(U;Y|A,X_1) - I(U;S|A,X_1)$ using a Gel'fand Pinsker scheme; i.e. choosing a codeword $U^n$ from a subcode-book $C(m_1,m_2)$, such that it is jointly typical with $(X_1^n,A^n,S^n)$, where$(X_1^n,A^n)$ were already decoded and are now used as side information at the encoder and decoder.

\begin{proof}
Fix a joint distribution of $P_{A,S,U,X_1,X_2,Y} = P_{X_1}P_{A|X_1}P_{S|A}P_{U|S,A,X_1}P_{X_2|X_1,S,U}P_{Y|X_1,X_2,S}$ where $P_{Y|X_1,X_2,S}$ is given by the channel.

\par{\it Code Construction}: For the uninformed encoder, generate $2^{nR_1}$ independent codewords, $X_1^n(m_1)$, $m_1\in\{1,2,...,2^{nR_1}\}$, where each element is i.i.d. $\sim\prod_{i=1}^n p(x_{1,i})$. For the informed encoder, generate $2^{n(R_1+R_2)}$ action sequences, $A^n(m_1,m_2) \sim\prod_{i=1}^n p(a_i|x_{1,i})$. In addition, for each set of messages, $(m_1,m_2)$, where $m_1 \times m_2 \in \{1,2,...,2^{n(R_1+R_2)}\}$, generate $2^{n(R_1+R_2)}$ bins. Next, generate randomly $2^{n\tilde{R}}$ codewords, $u(1),u(2),...,u(2^{n\tilde{R}})$, each according to $\sim\prod_{i=1}^n p(u_i|a_i,x_{1,i})$. Distribute the codewords uniformly among the bins. Now, each message set $(m_1,m_2)$ has a subcode-book denoted $c(m_1,m_2)$ of $2^{n(\tilde{R}-(R_1+R_2))}$ codewords. The code construction is illustrated in Fig. \ref{AchievadilityFig}.

\begin{figure}[h!]
\begin{center}
\begin{psfrags}
    \psfragscanon
    \psfrag{A}[][][0.9]{$2^{nR_1}$ codewords $X_1^n$}
    \psfrag{B}[][][0.9]{}
    \psfrag{C}[][][0.9]{\ \ \ \ \ \ \ \ $2^{n(\tilde{R}-(R_1+R_2))}$ codewords $U^n$}
    \psfrag{D}[][][0.9]{$2^{nR_1}-1$}
    \psfrag{E}[][][0.9]{$2^{nR_1}$}
    \psfrag{F}[][][0.9]{\ \ \ \ \ \ \ \ \ \ \ $2^{nR_2}$ Bins}
    \centerline{\includegraphics[scale = .8]{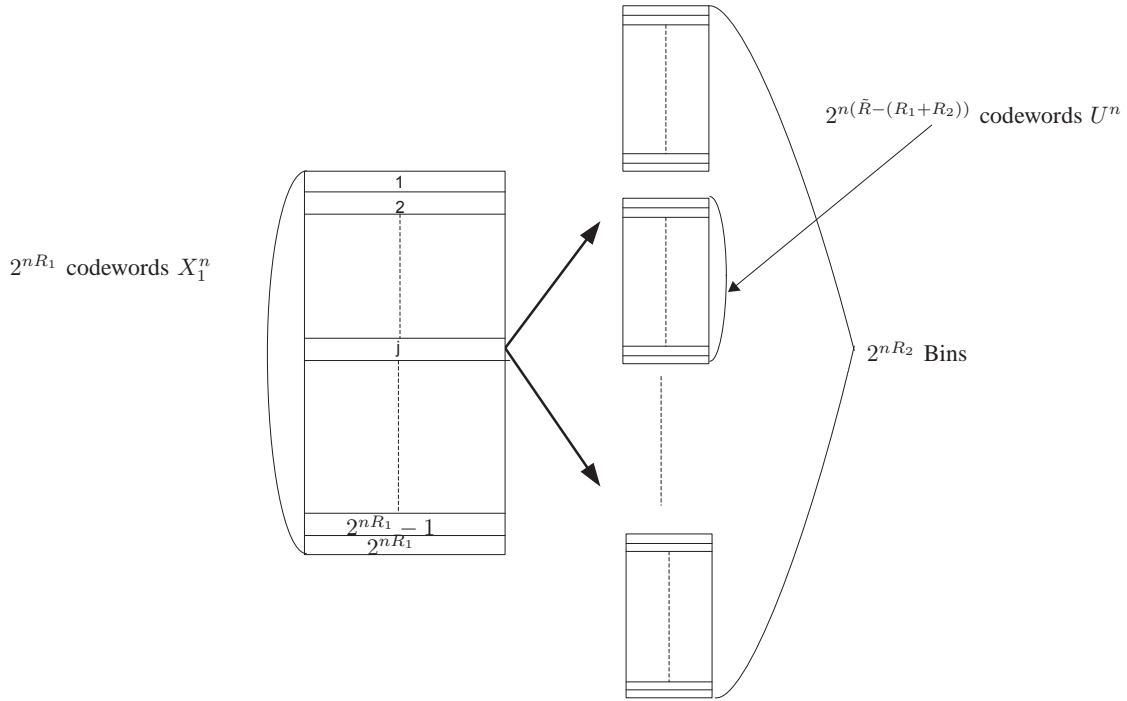}}
    \caption{Encoder 1 chooses the codeword $X_1^n(m_1)$ from its code-book of size $2^{nR_1}$. For each codeword, $X_1^n(m_1)$, we have $2^{nR_2}$ bins. We choose the second bin according to $m_2$ (notice that we have a total of $2^{nR_1}\times2^{nR_2}=2^{n(R_1+R_2)}$ bins, one for each message set, $(m_1,m_2)$). Now, we look in our bin, $(m_1,m_2)$, of $2^{n(\tilde{R}-(R_1+R_2))}$ codewords for a specific codeword, $U^n(m_1,m_2,k)$, such that it is jointly typical with $(X_1^n,S^n,A^n)$.} \label{AchievadilityFig}
\end{psfrags}
\end{center}
\end{figure}

\par {\it Encoding}: The uninformed encoder transmits $X_1^n(m_1)$. Next, the informed encoder chooses an action sequence $A^n(m_1,m_2)$. As a result, a state, $S^n$, is generated. Now, the informed encoder chooses a codeword $U^n(m_1,m_2,k)$ from bin $(m_1,m_2)$ with the smallest lexicographical order such that it is jointly typical with $(X_1^n,S^n,A^n)$, i.e.
\begin{equation}
(U^n(m_1,m_2,k),X_1^n(m_1),S^n,A^n)\in \mathcal{T}_{\epsilon}^{(n)}(U,X_1,S,A).
\end{equation}
If such a codeword, $U^n$, does not exist, namely, among the codewords in the bin none is jointly typical with $(X_1^n,S^n,A^n)$ , choose an arbitrary $U^n$ from the bin, $U^n(k)\in C(m_1,m_2)$ (in such a case the decoder will declare an error).
The input sequence to the channel, $X_2^n$, is according to $\sim p(x_2|x_1,s,u)$ and the encoder transmitter transmits $X_{2,i}$ at time $i \in [1,n]$.

\par {\it Decoding}: The Decoder looks for the smallest value of $\hat{m}_1,\hat{m}_2$ for which there exits a $\hat{k}$ such that:
\begin{equation}
(A^n(\hat{m}_1,\hat{m}_2),U^n(\hat{m}_1,\hat{m}_2,\hat{k}),X_1^n(\hat{m}_1),Y^n)\in \mathcal{T}_{\epsilon}^{(n)}
\end{equation}
If none or more than one such triplet is found, an error is declared. The estimated messages sent are $\hat{m}_1$ and $\hat{m}_2$.

\par{\bf Analysis of the probability of error:}\\
Without loss of generality, we can assume that messages $(m_1,m_2)=(1,1)$ were sent.\\
An error occurs in the following cases. We define the events:
\begin{eqnarray}
E_1 &=& \{\forall U^n\in C(1,1), (X_1^n,U^n,S^n,A^n)\notin \mathcal{T}_{\epsilon}^{(n)}(X_1,U,S,A)\},\\
E_2 &=& \{\forall U^n\in C(1,1), (A^n,X_1^n,U^n,Y^n)\notin \mathcal{T}_{\epsilon}^{(n)}(A,X_1,U,Y)\},\\
E_3 &=& \{\exists \hat{m}_2\neq 1: (A^n(1,\hat{m}_2),X_1^n(1),U^n(1,\hat{m}_2,k),Y^n)\in \mathcal{T}_{\epsilon}^{(n)}(A,X_1,U,Y)\},\\
E_4 &=& \{\exists \hat{m}_1\neq 1: (A^n(\hat{m}_1,1),X_1^n(\hat{m}_1),U^n(\hat{m}_1,1,k),Y^n)\in \mathcal{T}_{\epsilon}^{(n)}(A,X_1,U,Y)\},\\
E_5 &=& \{\exists \hat{m}_1\neq 1, \hat{m}_2\neq 1: (a^n(\hat{m}_1,\hat{m}_2),X_1^n(\hat{m}_1),U^n(\hat{m}_1,\hat{m}_2,k),Y^n)\in \mathcal{T}_{\epsilon}^{(n)}(A,X_1,U,Y)\}.
\end{eqnarray}
Then, by the union of events bound:
\begin{align*}
P_e^{(n)} &= \Pr(E_1\cup E_2\cup E_3\cup E_4\cup E_5)\\
          &\leq P(E_1)+P(E_2)+P(E_3)+P(E_4)+P(E_5).
\end{align*}
Now, let us find the probability of each event:\\

\begin{itemize}
\item For the first error, note that we have $2^{n(\tilde{R}-(R_1+R_2))}$ codewords in each bin. In addition, the probability of $(X_1^n(m_1),A^n(m_1,m_2))\in \mathcal{T}_{\epsilon}^{(n)}(X_1,A)$ is almost surely 1, and $S^n$ is distributed $\sim p(s^n|a^n,x_1^n)=\prod_{i=1}^n p(s_i|a_i)$. Therefore, if the number of codewords in each bin is bigger than $2^{nI(U;S|X_1,A)}$, namely $\tilde{R}-(R_1+R_2) > I(U;S|X_1,A)$, then, according to the covering lemma \cite{gamal2011network}, at least one codeword, $U^n$, is jointly typical with $(X_1^n,S^n,A^n)$ with high probability. Hence, if the number of codewords in each subcode-book is more then $2^{nI(U;S|X_1,A)}$, then the error $P(E_1)\rightarrow 0$ as $n\rightarrow \infty$. Thus,
    \begin{equation}
    \Pr\Big{(}(U^n(1,1,k),X_1^n,S^n,A^n(1,1))\in \mathcal{T}_{\epsilon}^{(n)}|X_1^n(1),A^n(1,1)\in \mathcal{T}_{\epsilon}^{(n)}(U,X_1,S,A)\Big{)}=2^{-nI(U;S|X_1,A)}, \nonumber
    \end{equation}
    therefore for $\tilde{R}-(R_1+R_2)>I(U;S|X_1,A)$ and $1\leq k \leq 2^{n(\tilde{R}-(R_1+R_2))}$, the probability
    \begin{equation}
    \Pr\Big{(}\cap_k(U^n(1,1,k),X_1^n(1),S^n,A^n(1,1))\notin \mathcal{T}_{\epsilon}^{(n)}|X_1^n(1),A^n(1,1)\in \mathcal{T}_{\epsilon}^{(n)}\Big{)} \rightarrow 0. \nonumber
    \end{equation}

\item Consider the second error:
\begin{equation}
P(E_2) = \Pr\Big{(}(A^n,X_1,U^n,Y^n)\notin \mathcal{T}_{\epsilon}^{(n)}\Big{)}. \nonumber
\end{equation}
As for the first event, if we have $\tilde{R}-(R_1+R_2)> I(U;S|X_1,A)$, then the probability of $(A^n,U^n,X_1^n,S^n)\in \mathcal{T}_{\epsilon}^{(n)}(U,X_1,S,A)$ approaches 1. Furthermore, due to the Markov Lemma \cite{Tom-Cover} and the Law of Large Numbers, as well as the fact that $Y^n$ is distributed according to $\sim P_{Y|X_1,X_2,S}$, the probability of $(A^n,U^n,S^n,X_1^n,X_2^n,Y^n)\in \mathcal{T}_{\epsilon}^{(n)}(U,X_1,S,A,Y)$  is almost surely 1. Hence, $P(E_2)\rightarrow 0$ as $n\rightarrow \infty$.\\

\item For the third error, consider:
\begin{equation}
P(E_3)= \Pr\Big{(}(A^n(1,\hat{m}_2),X_1^n(1),U^n(1,\hat{m}_2,k),Y^n)\in \mathcal{T}_{\epsilon}^{(n)}(A,X_1,U,Y)\Big{)}. \nonumber
\end{equation}
To bound the probability of this event, note that:
\begin{align}
\Pr\Big{(}(A^n(1,\hat{m}_2),X_1^n(1),U^n(1,\hat{m}_2,k),Y^n)\in \mathcal{T}_{\epsilon}^{(n)}(A,X_1,U,Y)|\hat{m}_2\neq 1\Big{)}&\leq 2^{-nI(A,U;Y|X_1)}\nonumber\\
&\stackrel{(*)}{\leq} 2^{-nI(U;Y|X_1)}.
\end{align}
Therefore, using the union bound:
\begin{equation}
\Pr\Big{(}\cup_k \{(A^n(1,\hat{m}_2),X_1^n(1),U^n(1,\hat{m}_2,k),Y^n)\in \mathcal{T}_{\epsilon}^{(n)}(A,X_1,U,Y)\}|\hat{m}_2\neq 1\Big{)}\leq 2^{-n(I(U;Y|X_1)-(\tilde{R}-(R_1+R_2)))}.
\end{equation}
Let us take $\tilde{R}-(R_1+R_2)=I(U;S|X_1,A)+\epsilon$. It follows that the probability where $\hat{m}_2\neq m_2$ for which there exists a $k$ such that $(A^n(1,j),X_1^n(1),U^n(1,j,k),Y^n)\in \mathcal{T}_{\epsilon}^{(n)}(A,X_1,U,Y)$ vanishes if we have $R_2< I(U;Y|X_1)-I(U;S|A,X_1)-\epsilon$.

\item For the forth error, consider:
\begin{equation}
P(E_4)= \Pr\Big{(}(A^n(\hat{m}_1,1),X_1^n(\hat{m}_1),U^n(\hat{m}_1,1,k),Y^n)\in \mathcal{T}_{\epsilon}^{(n)}(A,X_1,U,Y)\Big{)}. \nonumber
\end{equation}
To bound the probability of this event, note that:
\begin{align}
\Pr\Big{(}(A^n(\hat{m}_1,1),X_1^n(\hat{m}_1),U^n(\hat{m}_1,1,k),Y^n)\in \mathcal{T}_{\epsilon}^{(n)}(A,X_1,U,Y)|\hat{m}_1\neq 1\Big{)}&\leq 2^{-nI(A,X_1,U;Y)}\nonumber\\
&\stackrel{(*)}{\leq} 2^{-nI(X_1,U;Y)}.
\end{align}
Therefore, using the union bound:
\begin{align*}
\Pr\Big{(}\cup_k \{(A^n(\hat{m}_1,1),X_1^n(\hat{m}_1),U^n(\hat{m}_1,1,k),Y^n)\in \mathcal{T}_{\epsilon}^{(n)}(A,X_1,U,Y)\}|\hat{m}_1\neq 1\Big{)}&\leq 2^{-n(I(X_1,U;Y)-(\tilde{R}-(R_1+R_2)))}\\
                                                                                         &= 2^{-n(I(X_1,U;Y)-I(U;S|X_1,A)-\epsilon)}\\
                                                                                         &\stackrel{(**)}{=} 2^{-n(I(X_1,U;Y)-I(X_1,U;S|A)-\epsilon)},
\end{align*}
so it follows that the probability that there exists an $\hat{m}_1\neq m_1$ for which there exists a $k$ such that $(A^n(l,1),X_1^n(l),U^n(l,1,k),Y^n)\in \mathcal{T}_{\epsilon}^{(n)}(A,X_1,U,Y)$ vanishes if we have $R_1< I(X_1,U;Y)-I(X_1,U;S|A)$.
\\
\item For the fifth term, consider:
\begin{equation}
P(E_5)= \Pr\Big{(}(A^n(\hat{m}_1,\hat{m}_2),X_1^n(\hat{m}_1),U^n(\hat{m}_1,\hat{m}_2,k),Y^n)\in \mathcal{T}_{\epsilon}^{(n)}(A,X_1,U,Y)\Big{)}. \nonumber
\end{equation}
To bound the probability of this event note that:
\begin{align*}
\Pr\Big{(}(A^n(\hat{m}_1,\hat{m}_2),X_1^n(\hat{m}_1),U^n(\hat{m}_1,\hat{m}_2,k),Y^n)\in \mathcal{T}_{\epsilon}^{(n)}(A,X_1,U,Y)|\hat{m}_1\neq 1,\hat{m}_2\neq 1\Big{)}&\leq 2^{-nI(A,X_1,U;Y)}\\
                  &\stackrel{(*)}{\leq} 2^{-nI(X_1,U;Y)}.
\end{align*}
Therefore, using the union bound:
\begin{align*}
\Pr\Big{(}\cup_k \{(A^n(\hat{m}_1,\hat{m}_2),X_1^n(\hat{m}_1),U^n(\hat{m}_1,\hat{m}_2,k),Y^n)\in \mathcal{T}_{\epsilon}^{(n)}(A,X_1,U,Y)\}&|\hat{m}_1\neq 1,\hat{m}_2\neq 1\Big{)}\\
&\leq 2^{-n(I(X_1,U;Y)-(\tilde{R}-(R_1+R_2)))}\\
                                                                                                 &= 2^{-n(I(X_1,U;Y)-I(U;S|X_1,A)-\epsilon)}\\
                                                                                                 &\stackrel{(**)}{=} 2^{-n(I(X_1,U;Y)-I(X_1,U;S|A)-\epsilon)},
\end{align*}
so it follows that the probability that there exists an $\hat{m}_1\neq m_1$ and $\hat{m}_2\neq m_2$ for which there exists a $k$ such that $(A^n(l,j),X_1^n(l),U^n(l,j,k),Y^n)\in \mathcal{T}_{\epsilon}^{(n)}(A,X_1,U,Y)$ vanishes if we have $R_1+R_2< I(X_1,U;Y)-I(X_1,U;S|A)$.
\end{itemize}
$(*)$ Follows from {\it Corollary 1}.\\
$(**)$ Follows from the Markov chain $X_1-A-S$.

Combining the results, we have shown that $P(E)\rightarrow 0$ as $n\rightarrow \infty$ if
\begin{eqnarray}
R_2&\leq& I(U;Y|X_1)-I(U;S|X_1,A)\nonumber\\
R_1+R_2&\leq& I(X_1,U;Y)-I(X_1,U;S|A)\nonumber.
\end{eqnarray}
\par The above bound shows that the average probability of error, which, by symmetry, is equal to the probability for an individual pair of codewords, $(m_1,m_2)$, averaged over all choices of code-books in the random code construction, is arbitrarily small. Hence, there exists at least one code, $((2^{nR_1},2^{nR_2}),n)$, with an arbitrarily small probability of error.
\end{proof}


\section{Proof of Converse}\label{Converse}
In the previous section, we  proved the achievability part of  Theorem \ref{TheoremMain}. In this section, we provide the upper bound on the capacity region of the MAC with action-dependent state information at one encoder, i.e. we give the proof of the converse for  Theorem \ref{TheoremMain}.
\begin{proof}
Given an achievable rate, $(R_1,R_2)$, we need to show that there exists a joint distribution of the form $P_{X_1}P_{A|X_1}P_{S|A}P_{U|S,A,X_1}P_{X_1|X_1,S,U}P_{Y|X_1,X_2,S}$ such that
\begin{eqnarray}
R_2&\leq& I(U;Y|X_1)-I(U;S|A,X_1)\nonumber\\
R_1+R_2&\leq& I(X_1,U;Y)-I(X_1,U;S|A)\nonumber.
\end{eqnarray}
Since $(R_1,R_2)$ is an achievable pair-rate, there exists a code, $(n, 2^{nR_1} , 2^{nR_2})$, with a probability
of error, $P^{(n)}_e$, arbitrarily small. By Fano's inequality,
\begin{equation}
H(M_1,M_2|Y^n)\leq n(R_1+R_2)P^{(n)}_e+H(P^{(n)}_e)\triangleq \epsilon_n,
\end{equation}
so we can say that $\epsilon_n \rightarrow 0$ as $P^{(n)}_e \rightarrow 0$. Furthermore,
 \begin{equation}
H(M_2|M_1,Y^n)\leq H(M_1,M_2|Y^n)\leq \epsilon_n.
\end{equation}
To bound the rate $R_2$ consider:

\begin{align*}
    nR_2 &= H(M_2)\\
    & \stackrel{(a)}{=} H(M_2|M_1)\\
    & =H(M_2|M_1)-H(M_2|M_1,Y^n)+H(M_2|M_1,Y^n)\\
    & =I(M_2;Y^n|M_1)+H(M_2|M_1,Y^n)\\
    & \stackrel{(b)}{\leq}I(M_2;Y^n|M_1)+n\epsilon_n\\
    & \stackrel{(c)}{=}I(M_2;Y^n|M_1,X_1^n(M_1))+n\epsilon_n\\
    & \stackrel{(d)}{=}I(M_2;Y^n|M_1,X_1^n) - I(M_1,M_2;S^n|X_1^n,A^n)+n\epsilon_n\\
    & \stackrel{(e)}{=} \sum_{i=1}^n I(M_2;Y_i|M_1,X_1^n,Y^{i-1}) - I(M_1,M_2;S_i|X_1^n,S_{i+1}^n,A^n)+n\epsilon_n\\
    & \stackrel{(f)}{=} \sum_{i=1}^n I(M_2,S_{i+1}^n,A^n;Y_i|M_1,X_1^n,Y^{i-1})-I(S_{i+1}^n,A^n;Y_i|X_1^n,M_2,Y^{i-1}) \\&-I(M_1,M_2,Y^{i-1};S_i|X_1^n,S_{i+1}^n,A^n) + I(Y^{i-1};S_i,A^n|M_1,X_1^n,M_2,S_{i+1}^n,A^n)+n\epsilon_n\\
    & \stackrel{(g)}{=} \sum_{i=1}^n I(M_2,S_{i+1}^n,A^n;Y_i|M_1,X_1^n,Y^{i-1}) -I(M_1,M_2,Y^{i-1};S_i|X_1^n,S_{i+1}^n,A^n)+n\epsilon_n\\
    & \stackrel{(h)}{\leq} \sum_{i=1}^n H(Y_i|X_{1,i})-H(Y_i|X_1^n,M_1,M_2,S_{i+1}^n,Y^{i-1},A^n)\\
    &-[H(S_i|X_1^n,S_{i+1}^n,A^n)- H(S_i|X_1^n,M_1,M_2,Y^{i-1},S_{i+1}^n,A^n)]+n\epsilon_n\\
    & \stackrel{(i)}{=}\sum_{i=1}^n H(Y_i|X_{1,i})-H(Y_i|X_{1,i},U_i)-[H(S_i|X_{1,i},A_i) - H(S_i|X_{1,i},U_i,A_i)]+n\epsilon_n\\
    & = \sum_{i=1}^n I(Y_i;U_i|X_{1,i})-I(U_i;S_i|X_{1,i},A_i)+n\epsilon_n,
\end{align*}
where\\
$(a)$ follows from the fact that $M_1$ and $M_2$ are independent,\\
$(b)$ follows from Fano's inequality,\\
$(c)$ follows from the fact that $X_1^n(M_1)$ is a function of $M_1$,\\
$(d)$ follows from the fact that $(M_1,M_2,X_1^n)-A^n-S^n$ form a Markov chain,\\
$(e)$ follows from the chain rule,\\
$(f)$ follows from the fact that\\ $I(M_2;Y_i|M_1,X_1^n,Y^{i-1})=I(M_2,S_{i+1}^n,A^n;Y_i|M_1,X_1^n,Y^{i-1})-I(S_{i+1}^n,A^n;Y_i|M_1,M_2,X_1^n,Y^{i-1})$\\ (chain rule) and $I(M_1,M_2;S^n|M_1,X_1^n,S_{i+1}^n,A^n)=I(M_1,M_2,Y^{i-1};S_i|S_{i+1}^n,X_1^n,A^n) - I(Y^{i-1};S_i,A^n|M_2,S_{i+1}^n,M_1,X_1^n,A^n) $ (chain rule and Markov relation),\\
$(g)$ follows from the fact that\\
 $\sum_{i=1}^n I(S_{i+1}^n,A^n;Y_i|M_2,M_1,X_1^n,Y^{i-1})= \sum_{i=1}^n I(Y^{i-1};S_i,A^n|M_2,M_1,X_1^n,S_{i+1}^n,A^n)$\\
 due to the Csisz´ar sum identity,\\
$(h)$ follows from the definition of mutual information and from the fact that conditioning reduces entropy,\\
$(i)$ follows from the choice of $U_i = (X_1^{i-1},X_{1,i+1}^n,M_1,M_2,S_{i+1}^n,Y^{i-1},A^n)$ and the Markov chain $S_i-A_i-(S_{i+1}^n,A_{i-1},A_{i+1}^n,X^n)$.\\
Hence, we have:
\begin{equation}
R_2\leq \frac{1}{n}\sum_{i=1}^n [I(Y_i;U_i|X_{1,i})-I(U_i;S_i|X_{1,i},A_i)]+\epsilon_n. \label{eq2}
\end{equation}

\par To bound the sum of rates, $R_1+R_2$, consider:
\begin{align*}
    n(R_1+R_2) &= H(M_1,M_2)\\
    & =H(M_1,M_2)-H(M_1,M_2|Y^n)+H(M_1,M_2|Y^n)\\
    & =I(M_1,M_2;Y^n)+H(M_1,M_2|Y^n)\\
    & \stackrel{(a)}{\leq}I(M_1,M_2;Y^n)+n\epsilon_n\\
    & \stackrel{(b)}{=}I(M_1,X_1^n(M_1),M_2;Y^n)+n\epsilon_n\\
    & \stackrel{(c)}{=}I(X_1^n,M_1,M_2;Y^n) - I(X_1^n,M_1,M_2;S^n|A^n)+n\epsilon_n\\
    & \stackrel{(d)}{=} \sum_{i=1}^n I(X_1^n,M_1,M_2;Y_i|Y^{i-1})- I(X_1^n,M_1,M_2;S_i|S_{i+1}^n,A^n)+n\epsilon_n\\
    & \stackrel{(e)}{=} \sum_{i=1}^n I(X_1^n,M_1,M_2,S_{i+1}^n,A^n;Y_i|Y^{i-1})-I(S_{i+1}^n,A^n;Y_i|X_1^n,M_1,M_2,Y^{i-1}) \\& -I(X_1^n,M_1,M_2,Y^{i-1};S_i|S_{i+1}^n,A^n) + I(Y^{i-1};S_i,A^n|X_1^n,M_1,M_2,S_{i+1}^n,A^n)+n\epsilon_n\\
    & \stackrel{(f)}{=} \sum_{i=1}^n I(X_1^n,M_1,M_2,S_{i+1}^n,A^n;Y_i|Y^{i-1}) -I(X_1^n,M_1,M_2,Y^{i-1};S_i|S_{i+1}^n,A^n)+n\epsilon_n\\
    & \stackrel{(g)}{\leq} \sum_{i=1}^n H(Y_i)-H(Y_i|X_1^n,M_1,M_2,S_{i+1}^n,Y^{i-1},A^n)-[H(S_i|S_{i+1}^n,A^n)\\
    &- H(S_i|X_1^n,M_1,M_2,Y^{i-1},S_{i+1}^n,A^n)]+n\epsilon_n\\
    & \stackrel{(h)}{=}\sum_{i=1}^n H(Y_i)-H(Y_i|X_{1,i},U_i)-[H(S_i|A_i) - H(S_i|U_i,X_{1,i},A_i)]+n\epsilon_n\\
    & = \sum_{i=1}^n I(Y_i;X_{1,i},U_i)-I(X_{1,i},U_i;S_i|A_i)+n\epsilon_n,
\end{align*}
where\\
$(a)$ follows from Fano's inequality,\\
$(b)$ follows from the fact that $X_1^n(M_1)$ is a function of $M_1$,\\
$(c)$ follows from the fact that $(X_1^n,M_1,M_2)-A^n-S^n$ form a Markov chain,\\
$(d)$ follows from the chain rule,\\
$(e)$ follows from the fact that $I(X_1^n,M_1,M_2;Y_i|Y^{i-1})=I(X_1^n,M_1,M_2,S_{i+1}^n,A^n;Y_i|Y^{i-1})-I(S_{i+1}^n,A^n;Y_i|X_1^n,M_1,M_2,Y^{i-1})$ (chain rule) and $I(X_1^n,M_1,M_2;S^n|S_{i+1}^n,A^n)=I(X_1^n,M_1,M_2,Y^{i-1};S_i|S_{i+1}^n,A^n) - I(Y^{i-1};S_i,A^n|X_1^n,M_1,M_2,S_{i+1}^n,A^n) $ (chain rule and Markov chain),\\
$(f)$ follows from the fact that\\
 $\sum_{i=1}^n I(S_{i+1}^n,A^n;Y_i|X_1^n,M_1,M_2,Y^{i-1})= \sum_{i=1}^n I(Y^{i-1};S_i,A^n|X_1^n,M_1,M_2,S_{i+1}^n,A^n)$\\
 due to the Csisz´ar sum identity,\\
$(g)$ follows from the definition of mutual information and from the fact that conditioning reduces entropy,\\
$(h)$ follows from the choice of $U_i = (X_1^{i-1},X_{1,i+1}^n,M_1,M_2,S_{i+1}^n,Y^{i-1},A^n)$ and the Markov chain $S_i-A_i-(S_{i+1}^n,A_{i-1},A_{i+1}^n)$.\\
Hence, we have:
\begin{equation}
R_1+R_2\leq \frac{1}{n}\sum_{i=1}^n [I(Y_i;X_{1,i},U_i)-I(X_{1,i},U_i;S_i|A_i)]+\epsilon_n. \label{eq1}
\end{equation}

\par The expressions in (\ref{eq1}) and (\ref{eq2}) are the average of the mutual informations calculated from the empirical
distribution in column $i$ of the code-book. We can rewrite these equations with the new variable, $Q$, where $Q = i\in
\{1, 2, ..., n\}$ with probability $\frac{1}{n}$. The equations become:

\begin{align*}
R_2 &\leq \frac{1}{n}\sum_{i=1}^n [I(Y_i;U_i|X_{1,i})-I(U_i;S_i|X_{1,i},A_i)]+\epsilon_n\\
&=\frac{1}{n}\sum_{i=1}^n [I(Y_Q;U_Q|X_{1,Q},Q=i)-I(U_Q;S_Q|X_{1,Q},A_Q,Q=i)]+\epsilon_n\\
&=I(Y_Q;U_Q|X_{1,Q},Q)-I(U_Q;S_Q|X_{1,Q},A_Q,Q)+\epsilon_n\\
&= I(Y_Q;U_Q|X_{1,Q},Q)-I(U_Q,Q;S_Q|X_{1,Q},A_Q)+I(Q;S_Q|X_{1,Q},A_Q)+\epsilon_n\\
&\leq I(Y_Q;U_Q,Q|X_{1,Q})-I(U_Q,Q;S_Q|X_{1,Q},A_Q)+\epsilon_n,
\end{align*}

and similarly,
\begin{align*}
R_1+R_2&\leq \frac{1}{n}\sum_{i=1}^n [I(Y_i;X_{1,i},U_i)-I(X_{1,i},U_i;S_i|A_i)]+\epsilon_n\\
&= \frac{1}{n}\sum_{i=1}^n [I(Y_Q;X_{1,Q},U_Q|Q=i)-I(X_{1,Q},U_Q;S_Q|A_Q,Q=i)]+\epsilon_n\\
&= I(Y_Q;X_{1,Q},U_Q|Q)-I(X_{1,Q},U_Q;S_Q|A_Q,Q)+\epsilon_n\\
&= I(Y_Q;X_{1,Q},U_Q|Q)-I(X_{1,Q},U_Q,Q;S_Q|A_Q)+I(Q;S_Q|A_Q)+\epsilon_n\\
&\leq I(Y_Q;X_{1,Q},U_Q,Q)-I(X_{1,Q},U_Q,Q;S_Q|A_Q)+\epsilon_n,
\end{align*}
where the last step follows from the fact that $I(Y_Q;X_{1,Q},U_Q|Q)\leq I(Y_Q;X_{1,Q},U_Q,Q)$ and the stationarity of $S_i$.
\par Let us denote $X_1 \triangleq X_{1Q}$, $U \triangleq (U_Q,Q)$, $Y \triangleq Y_Q$, $S \triangleq S_Q$, $A \triangleq A_Q$. Now, taking the limit as $n\rightarrow \infty$, $P_e^{(n)}\rightarrow 0$, we conclude:
\begin{eqnarray}
R_2&\leq&I(U;Y|X_1)-I(U;S|X_1,A)\nonumber\\
R_1+R_2&\leq& I(X_1,U;Y)-I(X_1,U;S|A).\nonumber
\end{eqnarray}

\par Next, we want to prove the cardinality bound on $U$. Using the support lemma introduced in \cite{Csiszar-Korner}, $U$ needs to contain $|\mathcal{A}||\mathcal{S}||\mathcal{X}_1||\mathcal{X}_2|-1$ elements in order to preserve the joint distribution $P_{A,S,X_1,X_2}$. In addition, we need two more elements to preserve the expressions $H(Y|U,X_1)$ and $H(S|X_1,A,U)$. Further, from the markov relation $U-(X_1,X_2,S)-Y$, we can determine that the joint distribution $P_{Y,S,X_1,X_2}$ is also preserved. Therefore the cardinality of $U$ is bounded by $|\mathcal{A}||\mathcal{S}||\mathcal{X}_1||\mathcal{X}_2|+1$.

\par Finally, we will prove the following Markov relations to complete the converse:\\
\begin{enumerate}
\item $P_{S|A,X_1}(s|a,x_1)=P_{S|A}(s|a)$.\\
\begin{align*}
p(s|a,x_1)  &=\sum_{(m_1,m_2)} p(s,m_1,m_2|a,x_1)\\
            &\stackrel{(a)}{=} \sum_{(m_1,m_2)}p(m_1,m_2|a,x_1)p(s|a,x_1,m_1,m_2)\\
            &\stackrel{(b)}{=} \sum_{(m_1,m_2)}p(m_1,m_2|a,x_1)p(s|a,m_1,m_2)\\
            &\stackrel{(c)}{=} \sum_{(m_1,m_2)}p(m_1,m_2|a,x_1)p(s|a)\\
            &= p(s|a)\sum_{(m_1,m_2)}p(m_1,m_2|a,x_1) = p(s|a),\\
\end{align*}

where\\
$(a)$ follows from the chain rule,\\
$(b)$ follows from the fact that $X_1=f(M_1)$,\\
$(c)$ follows from the channel model that states that $(M_1,M_2)-A-S$ form a Markov chain.\\
\par Thus, we can say that the following Markov, $X_1-(M_1,M_2)-A-S$, holds. Therefore, we conclude that $P_{S|A,X_1}(s|a,x_1)=P_{S|A}(s|a)$.
\\

\item $P_{Y|X_1,A,S,U,X_2}(y|x_1,a,s,u,x_2)=P_{Y|X_1,X_2,S}(y|x_1,x_2,s)$\\
Follows from the fact that the channel output at any time, $i$, is assumed to depend only on the channel inputs
and state at time $i$.\\

\item $P_{X_2|X_1,S,U,A}(x_2|x_1,s,u,a)=P_{X_2|X_1,S,U}(x_2|x_1,s,u)$\\
This is due to the fact that we take the auxiliary random variable $U$ to be $U_i = (X_1^{i-1},X_{1,i+1}^n,M_1,M_2,S_{i+1}^n,Y^{i-1},A^n)$.
\end{enumerate}
With this we complete the proof.
\end{proof}


\section{The Gaussian Action-MAC}\label{GaussMac}


In this section, we examine the Gaussian channel setting of our Action-MAC. We present the channel model and the power constraints and obtain a closed formula for the capacity region. Using similar steps to those used in the proof of the Gaussian Action-MAC, we also derive the capacity expression for the Gaussian action-dependent point-to-point channel. The point-to-point model was introduced in \cite{Tsachy-Weissman} and was left open. We proceed to discuss the results of the Gaussian Action-MAC and observe that the formula for the capacity region found contains previously known results, such as the capacity region for the Gaussian Generalized Gel'fand Pinsker (GGP) MAC found in \cite{Baruch-Shamai-Verd} as well as the new action-dependent point-to-point result.
Furthermore, in Appendix \ref{AppendixPTPGauss} we give an alternative proof for the capacity of the action-dependent point-to-point channel, using a one-to-one correspondence to the GGP MAC \cite{Baruch-Shamai-Verd} with only a common message.

Section \ref{GaussianChannelModel} presents the channel model, the relations between the channel variables, the power constraints and the main theorem. In Section \ref{GaussainConverse}, we find an upper bound on the capacity region of the Gaussian channel. Section \ref{GaussianDirect} shows that the upper bound found can be achieved by taking particular distributions of our random variables, thus attaining that our upper bound is, indeed, tight. We conclude with some remarks on different results that can be derived from our capacity region.


\subsection{Capacity and Channel Model}\label{GaussianChannelModel}
Consider the following setting. The channel probability is given by the following relation between $X_{1,i},X_{2,i},S_i$ and $Y_i$
\begin{equation}
Y_i=X_{1,i}+X_{2,i}+S_i+Z_i,
\end{equation}
where
 \begin{itemize}
 \item $Z^n$ is an i.i.d Gaussian noise process with zero-mean and $E[Z_i^2]=N$, independent of $X$ and $S$.
 \item The process, $S^n$, is combined out of $A^n(M_1,M_2)$ and $W^n$, i.e. $S^n=A^n(M_1,M_2)+W^n$. Therefore, we can look at the channel output, $Y_i$, as:
     \begin{equation}
     Y_i=X_{1,i}+X_{2,i}+A_i+W_i+Z_i.\label{GaussianMac1}
     \end{equation}
 \item $W^n$ is an i.i.d Gaussian noise process with zero-mean and $E[W_i^2]=Q$, independent of $Z^n$ and $X_1^n,A^n$.
 \item The actions are constrained by
\begin{equation}
\frac{1}{n}\sum_{i=1}^n A_i\leq P_A.\label{PowerCon1}
\end{equation}
\item The power constraints on the channel inputs are:
\begin{eqnarray}
\frac{1}{n}\sum_{i=1}^n X_{1,i}^2\leq P_1,\ \ \ \ \ \ \ \ \frac{1}{n}\sum_{i=1}^n X_{2,i}^2\leq P_2.\label{PowerCon2}
\end{eqnarray}
\end{itemize}
\par We obtain the capacity region for this channel in the following theorem.

\begin{theorem}\label{GaussianTheorem}
The capacity region of the Gaussian action-dependent MAC, under the power constraints (\ref{PowerCon1})-(\ref{PowerCon2}) is given by the union of rate pairs satisfying
{\footnotesize{
\begin{eqnarray}
R_2\leq\frac{1}{2}\log{\frac{\Big{(}N+P_2+P_A+Q-P_2\rho_{12}^2-P_A\rho_{1A}^2+2\sqrt{P_2P_A}\rho_{2A}-2\sqrt{P_2P_A}\rho_{12}\rho_{1A}+2\sqrt{P_2Q}\rho_{2W}\Big{)}\Big{(}N(\rho_{1A}^2-1)-P_2\Delta\Big{)}}{N\Big{(}(\rho_{1A}^2-1)(N+Q+P_2\rho_{2W}^2+2\sqrt{P_2Q}\rho_{2W})-P_2\Delta\Big{)}}}\label{RegionTheorem1}\\
R_1+R_2\leq\frac{1}{2}\log{\frac{\Big{(}N+P_1+P_2+P_A+Q+2\sqrt{P_1P_2}\rho_{12}+2\sqrt{P_1P_A}\rho_{1A}+2\sqrt{P_2P_A}\rho_{2A}+2\sqrt{P_2Q}\rho_{2W}\Big{)}\Big{(}N(\rho_{1A}^2-1)-P_2\Delta\Big{)}}{N\Big{(}(\rho_{1A}^2-1)(N+Q+P_2\rho_{2W}^2+2\sqrt{P_2Q}\rho_{2W})-P_2\Delta\Big{)}}},\label{RegionTheorem2}
\end{eqnarray}
}}
for some $\rho_{12}\in[-1,1]$, $\rho_{1A}\in[-1,1]$, $\rho_{2A}\in[-1,1]$, $\rho_{2W}\in[-1,1]$ where
\begin{eqnarray}
\Delta = 1 -\rho_{12}^2-\rho_{1A}^2-\rho_{2A}^2-\rho_{2W}^2 + \rho_{1A}^2\rho_{2W}^2+2\rho_{1A}\rho_{2A}\rho_{12},
\end{eqnarray} such that
\begin{eqnarray}\label{GaussConstraint}
\Delta\geq 0.
\end{eqnarray}
\end{theorem}

\par The proof of this theorem is obtained in three stages. First, in the converse subsection, we state two lemmas that show that the capacity region (\ref{CapacityRegion}) given in Theorem \ref{TheoremMain} is upper bounded by the region
\begin{eqnarray}
R_2&\leq&I(A;Y|X_1)+h(X_2|X_1,A,W)-h(X_2-\hat{X}_2^{\text{lin}}(X_1,A,W,X_2+Z))\\
R_1+R_2&\leq&I(A,X_1;Y)+h(X_2|X_1,A,W)-h(X_2-\hat{X}_2^{\text{lin}}(X_1,A,W,X_2+Z)),
\end{eqnarray}
where $(X_1,X_2,A,W,Z,Y)$ are jointly Gaussian and $\hat{X}_2^{\text{lin}}(X_1,A,W,X_2+Z)$ is the optimal MMSE estimator given the measurements $(X_1,A,W,X_2+Z)$. Second, the exact expressions of the upper bound for our region are calculated as functions of $\sigma_{X_1}^2$, $\sigma_{X_2}^2$ and $\sigma_{A}^2$. Next, we show that replacing $\sigma_{X_1}^2$, $\sigma_{X_2}^2$ and $\sigma_{A}^2$ with $P_1$, $P_2$ and $P_A$, respectively, further increases the region.

This solution introduces the constraint (\ref{GaussConstraint}) which follows from the fact the allowable values of the covariances $E[X_1A]=\sigma_{1A}$, $E[X_1X_2]=\sigma_{12}$, $E[X_2A]=\sigma_{2A}$ and $E[X_2W]=\sigma_{2W}$ are such that the covariance matrix
\begin{equation}\Lambda=
\left(
  \begin{array}{ccccc}
    P_1 & \sigma_{12} & \sigma_{1A} & 0 & 0 \\
    \sigma_{12} & P_2 & \sigma_{2A} & \sigma_{2W} & 0 \\
    \sigma_{1A} & \sigma_{2A} & P_A & 0 & 0 \\
    0 & \sigma_{2W} & 0 & Q & 0 \\
    0 & 0 & 0 & 0 & N \\
  \end{array}
\right),
\end{equation}
satisfies the nonnegative-definiteness condition, i.e.
\begin{eqnarray}
\det{(\Lambda)}=\sigma_{1A}^2 \sigma_{2W}^2 N P_1 P_A +2\sigma_{12}\sigma_{1A}\sigma_{2A}NQ-\sigma_{2A}^2N P_1 Q -\sigma_{12}^2 N P_A Q +NP_1 P_2 P_A Q\geq 0,
\end{eqnarray}
or, equivalently, as a function of $\rho_{12},\rho_{1A},\rho_{2A}$ and $\rho_{2W}$
\begin{eqnarray}
\Delta = 1 -\rho_{12}^2-\rho_{1A}^2-\rho_{2A}^2-\rho_{2W}^2 + \rho_{1A}^2\rho_{2W}^2+2\rho_{1A}\rho_{2A}\rho_{12}\geq 0,
\end{eqnarray}
where
\begin{eqnarray}
\rho_{1A}=\frac{\sigma_{1A}}{\sqrt{P_1P_A}} \ \ \ \ \ \rho_{12}=\frac{\sigma_{12}}{\sqrt{P_1P_2}} \ \ \ \ \ \rho_{2A}=\frac{\sigma_{2A}}{\sqrt{P_2P_A}} \ \ \ \ \ \rho_{2W}=\frac{\sigma_{2W}}{\sqrt{P_2Q}}.
\end{eqnarray}
Finally, in the direct part, we show that this region is, indeed, achievable.

 \par The capacity region for the Gaussian Action-MAC setting is plotted in Fig. \ref{RegionVsPa} for various values of $P_A$. It is particularly interesting to notice that by taking $P_A=0$ we derive the rate region for the GGP MAC plotted in \cite{Baruch-Shamai-Verd}. We can interpret this as taking $A=0$, i.e. disregarding the action from our Gaussian channel model reduces our model to the Gaussian GGP MAC. In addition, we also see that the rate region increases with $P_A$. However, we see that the region only grows logarithmically fast as a function of $P_A$, a fact that we can also derive from (\ref{RegionTheorem1})-(\ref{RegionTheorem2}).
\vspace{-5mm}

 \begin{figure}[h!]
    \begin{center}
        \begin{psfrags}
        \psfragscanon
        \psfrag{A}[][][0.75]{$P_A=0$}
        \psfrag{B}[][][0.75]{$P_A=1$}
        \psfrag{C}[][][0.75]{$P_A=2$}
        \psfrag{D}[][][0.75]{$P_A=3$}
        \psfrag{E}[][][0.75]{$P_A=4$}
        \psfrag{F}[][][0.75]{$P_A=5$}
            \centerline{\includegraphics[scale = .45]{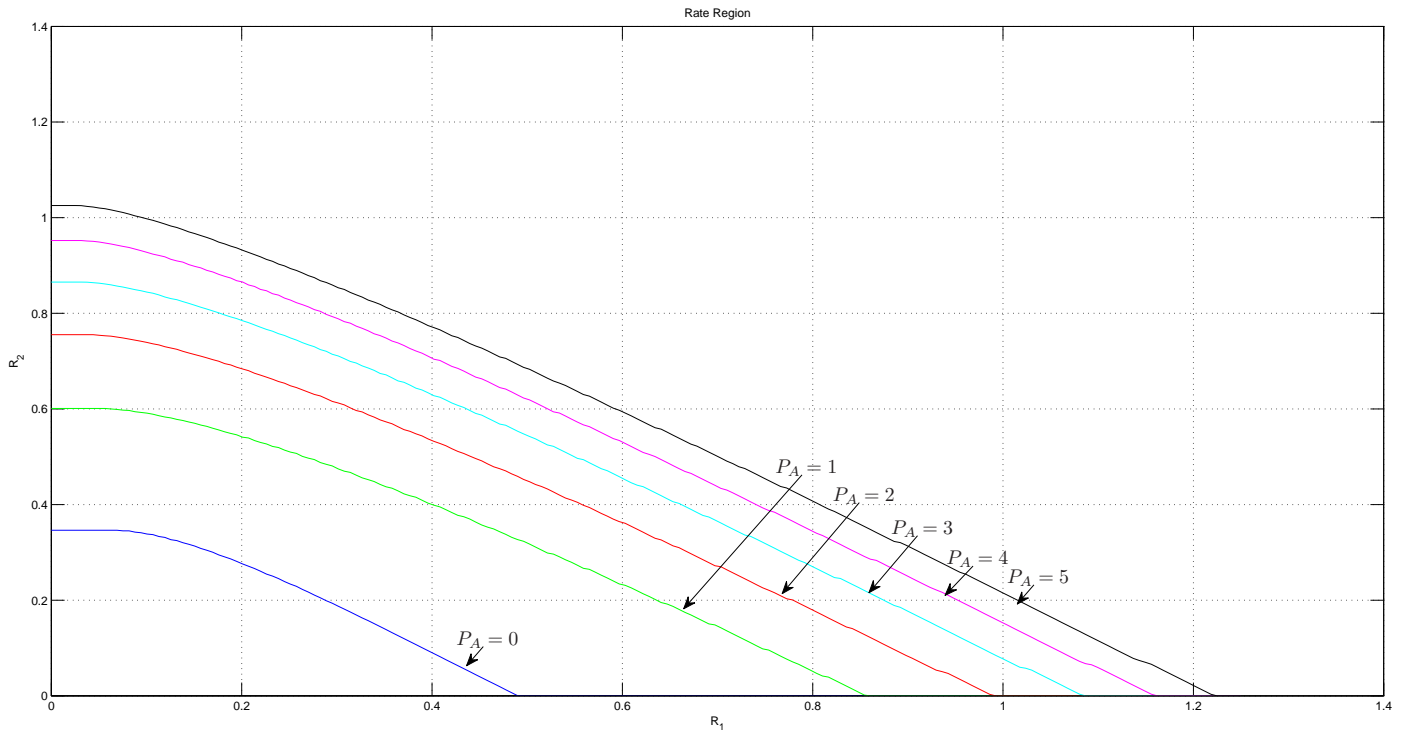}}
            \caption{The capacity region for various values of $P_A$ and $P_1=0.5$, $P_2=1$, $Q=1.5$, $N=1$. The smallest region corresponds to $P_A=0$ and is the exact region of the GGP MAC plotted in \cite{Baruch-Shamai-Verd}. The next regions correspond to $P_A=1,2,3,4,5$, where we can see that the regions grow grow as $P_A$ increases, but by smaller increments.} \label{RegionVsPa}
            \psfragscanoff
        \end{psfrags}
     \end{center}
 \end{figure}

\vspace{-10mm}
\subsection{Converse}\label{GaussainConverse}
In this section, we present a converse to the upper bound for the Gaussian setting of the Action-MAC. We start by finding an upper bound to the region found in Theorem \ref{TheoremMain} and show that it suffices to consider only random variables $(X_1,X_2,A,W,Z,Y)$ that are jointly Gaussian. Next, we calculate the expressions for the upper bound in terms of $E[X_1^2]$, $E[X_2^2]$ and $E[A^2]$. We show that replacing the last terms with $P_1$, $P_2$ and $P_A$, respectively, further increases the upper bound. Furthermore, we add the nonnegative-definite constraint on the covariance matrix of our variables such that we obtain an equivalent problem in terms of $P_1$, $P2$, $P_A$, $\rho_{12}$, $\rho_{1A}$, $\rho_{2A}$ and $\rho_{2W}$. We conclude by obtaining a set of equalities that give us a simpler representation for our upper bound. In the direct part, we will see that this representation for the upper bound is achievable.

We begin by presenting the next lemma, where we show that the following region provides an upper bound to the region given in Theorem \ref{TheoremMain}:
\begin{lemma}\label{GaussaianLemma1}
The closure of the convex hull of the set of rate pairs satisfying
\begin{eqnarray}\label{GaussianUpperBound}
R_2\leq I(A;Y|X_1)+I(X_2;Y|X_1,A,W)\nonumber\\
R_1+R_2\leq I(A,X_1;Y)+I(X_2;Y|X_1,A,W),
\end{eqnarray}
for some joint probability distribution of the form
\begin{equation}
P_{A,S,X_1,X_2,Y} = P_{X_1}P_{A|X_1}P_{S|A}P_{X_2|X_1,S,A}P_{Y|X_1,X_2,S},\nonumber
\end{equation}
is an outer bound on the capacity region found in Theorem \ref{TheoremMain}.
\end{lemma}
\begin{proof}
Recall the equivalent capacity region for Theorem \ref{TheoremMain}, (\ref{Corollary1}), given in Corollary \ref{intuitiveRegion}.
For the term $R_1+R_2\leq I(X_1,A,U;Y)+I(X_1,U;S|A)$ in (\ref{Corollary1}) we have:
\begin{align}
I(X_1,A,U;Y)+I(X_1,U;S|A)&= I(X_1,A,U;Y)-I(X_1,U;W+A|A)\nonumber\\
                           &\stackrel{(a)}{=}I(X_1,A,U;Y)-I(X_1,A,U;W)\nonumber\\
                           &= I(X_1,A;Y)+I(U;Y|X_1,A)-I(X_1,A;W)-I(U;W|X_1,A)\nonumber\\                   &\stackrel{(b)}{=} I(X_1,A;Y)+I(U;Y|X_1,A)-I(U;W|X_1,A)\nonumber\\
                           &= I(X_1,A;Y)-h(U|X_1,A,Y)+h(U|X_1,A,W)\nonumber\\
                           &= I(X_1,A;Y)+I(U;Y|X_1,A,W)-I(U;W|X_1,A,Y)\nonumber\\
                           &\leq I(X_1,A;Y)+I(U;Y|X_1,A,W)\nonumber\\
                           &=I(X_1,A;Y)+h(Y|X_1,A,W)-h(Y|X_1,A,W,U)\nonumber\\
                           &\stackrel{(c)}{\leq}I(X_1,A;Y)+h(Y|X_1,A,W)-h(Y|X_1,X_2,A,W,U)\nonumber\\                           &\stackrel{(d)}{=}I(X_1,A;Y)+I(X_2;Y|X_1,A,W),\nonumber
\end{align}
where\\
$(a)$ is due to the fact that $W$ and $A$ are independent of each other,\\
$(b)$ is due to the fact that $W$ is independent of $(X_1,A)$,\\
$(c)$ is due to the fact that conditioning reduces entropy,\\
$(d)$ is due to the properties of the channel,
\par The proof for $R_2\leq I(A;Y|X_1)+I(X_2;Y|X_1,A,W)$ is obtained in a similar manner to that above, except that for $R_2$ we have the term $I(A;Y|X_1)$ instead of $I(X_1,A;Y)$.
\end{proof}
\par Now, that we have an upper bound (\ref{GaussianUpperBound}) for the capacity region of Theorem \ref{TheoremMain}, we show in Lemma \ref{GaussaianLemma2} that choosing jointly Gaussian $(X_1,X_2,A,W,Z,Y)$ further increases the region.
\begin{lemma}\label{GaussaianLemma2}
The region given in (\ref{GaussianUpperBound}) increases by taking $(X_1,X_2,A,W,Z,Y)$ to be jointly Gaussian.
\end{lemma}
\begin{proof}
Let us take the expression bounding $R_1+R_2$, i.e.  $I(A,X_1;Y)+I(X_2;Y|X_1,A,W)$:
\begin{align}
I(A,X_1;Y)+I(X_2;Y|X_1,A,W)&= h(Y)-h(Y|X_1,A)+h(Y|X_1,A,W)-h(Y|X_1,X_2,A,W)\nonumber\\
                           &= h(Y)-I(W;Y|X_1,A)-h(Z)\nonumber\\
                           &\stackrel{(a)}{=} h(Y)+h(W|Y,X_1,A)-h(W)-h(Z),
\end{align}
where $(a)$ follows from the fact that $W$ is independent of $(A,X_1)$.
\par We invoke the maximum differential entropy lemma \cite{gamal2011network}, which states that the differential entropy, $h(X^n|Y^n)$, for a pair of random vectors, $(X^n,Y^n)\sim f(x^n,y^n)$, with covariance matrices, $K_X$ and $K_Y$, is maximized for jointly Gaussian ($X^n,Y^n)$. Hence, the differential entropies of $h(Y)$ and $h(W|Y,X_1,A)$ are maximized if we take $(X_1,X_2,A,W,Z,Y)$ to be jointly Gaussian.
\par In the same manner, for the expression bounding $R_2$, i.e. $I(A;Y|X_1)+I(X_2;Y|X_1,A,W)$, we have
\begin{equation}
I(A;Y|X_1)+I(X_2;Y|X_1,A,W)=h(Y|X_1)+h(W|Y,X_1,A)-h(W)-h(Z).
\end{equation}
Studying the last expression again reveals that it is maximized by taking $(X_1,X_2,A,W,Z,Y)$ to be jointly Gaussian.
\end{proof}
Now we would like to calculate the expressions for fixed second moments
\begin{eqnarray}
E[X_1^2]=\sigma_{X_1}^2\leq P_1 \ \ \ \ E[X_2^2]=\sigma_{X_2}^2\leq P_2 \ \ \ \ E[A^2]=\sigma_{A}^2\leq P_A\nonumber\\
E[X_1X_2]=\sigma_{12} \ \ \ \ E[X_1A]=\sigma_{1A} \ \ \ \ E[X_2A]=\sigma_{2A} \ \ \ \ E[X_2W]=\sigma_{2W}.
\end{eqnarray}
Furthermore, we have
\begin{eqnarray}
\sigma_Y^2&=&\sigma_{X_1}^2+\sigma_{X_2}^2+\sigma_{A}^2+Q+N+2\sigma_{12}+2\sigma_{1A}+2\sigma_{2A}+2\sigma_{2W},\nonumber\\
\sigma_{Y|X_1}^2&=&\sigma_Y^2-\frac{(\sigma_{X_1}^2+\sigma_{12}+\sigma_{1A})^2}{\sigma_{X_1}^2},\nonumber\\
\sigma_{W|(Y,X_1,A)}^2&=&\Sigma_{11}-\Sigma_{12}\Sigma_{22}^{-1}\Sigma_{21},
\end{eqnarray}
where
\begin{eqnarray}
\Sigma_{11}&=&E[W^2]=Q,\\
\Sigma_{12}&=&E[W\cdot(Y,X_1,A)]=\Big{(}Q+\sigma_{2W}, 0 , 0\Big{)},\\
\Sigma_{21}&=&E[W\cdot (Y,X_1,A)^T]=\Big{(}Q+\sigma_{2W}, 0 , 0\Big{)}^T,\\
\Sigma_{22}&=&E[(Y,X_1,A)^T\cdot (Y,X_1,A)],\nonumber\\
&=&{\scriptsize{\left(
                                                   \begin{array}{ccc}
                                                     \sigma_y^2 & \sigma_{X_1}^2+\sigma_{1A}+\sigma_{12} & \sigma_{A}^2+\sigma_{1A}+\sigma_{2A} \\
                                                      \sigma_{X_1}^2+\sigma_{1A}+\sigma_{12} & \sigma_{X_1}^2 & \sigma_{1A} \\
                                                     \sigma_{A}^2+\sigma_{1A}+\sigma_{2A} & \sigma_{1A} & \sigma_{A}^2 \\
                                                   \end{array},
                                                 \right)}}\nonumber.
\end{eqnarray}
Therefore
\begin{align}
\hspace{-4.7cm}R_2&\leq h(Y|X_1)+h(W|Y,X_1,A)-h(W)-h(Z)\nonumber\\
    \hspace{-4.5cm}&=\frac{1}{2}\log{(\frac{\sigma_{Y|X_1}^2\sigma_{W|Y,X_1,A}^2}{QN})}\nonumber
\end{align}
{\footnotesize{
\begin{align}
     \hspace{0.6cm}&=\frac{1}{2}\log{\frac{\Big{(}N\hspace{-0.5mm}+\hspace{-0.5mm}\sigma_{X_2}^2\hspace{-0.5mm}+\hspace{-0.5mm}\sigma_{A}^2\hspace{-0.5mm}+\hspace{-0.5mm}Q\hspace{-0.5mm}-\hspace{-0.5mm}\sigma_{X_2}^2\rho_{12}^2\hspace{-0.5mm}-\hspace{-0.5mm}\sigma_{A}^2\rho_{1A}^2\hspace{-0.5mm}+\hspace{-0.5mm}2\sqrt{\sigma_{X_2}^2\sigma_{A}^2}\rho_{2A}\hspace{-0.5mm}-\hspace{-0.5mm}2\sqrt{\sigma_{X_2}^2\sigma_{A}^2}\rho_{12}\rho_{1A}\hspace{-0.5mm}+\hspace{-0.5mm}2\sqrt{\sigma_{X_2}^2Q}\rho_{2W}\Big{)}\Big{(}N(\rho_{1A}^2-1)-\sigma_{X_2}^2\Delta\Big{)}}{N\Big{(}(\rho_{1A}^2-1)(N+Q+\sigma_{X_2}^2\rho_{2W}^2+2\sqrt{\sigma_{X_2}^2Q}\rho_{2W})-\sigma_{X_2}^2\Delta\Big{)}}}\nonumber
\end{align}}}
\begin{align}
\hspace{-4.7cm}&=\Gamma(\sigma_{X_2}^2,\sigma_{A}^2,Q,N,\rho_{12},\rho_{1A},\rho_{2A},\rho_{2W})
\end{align}
and
\begin{align}
\hspace{-5.1cm}R_1+R_2&\leq h(Y)+h(W|Y,X_1,A)-h(W)-h(Z)\nonumber\\
        \hspace{-4.5cm}&=\frac{1}{2}\log{(\frac{\sigma_{Y}^2\sigma_{W|Y,X_1,A}^2}{QN})}\nonumber
\end{align}
{\footnotesize{
\begin{align}
        \hspace{1.3cm}&=\frac{1}{2}\log{\frac{\Big{(}N\hspace{-0.5mm}+\hspace{-0.5mm}\sigma_{X_1}^2\hspace{-0.5mm}+\hspace{-0.5mm}\sigma_{X_2}^2\hspace{-0.5mm}+\hspace{-0.5mm}\sigma_{A}^2\hspace{-0.5mm}+\hspace{-0.5mm}Q\hspace{-0.5mm}+\hspace{-0.5mm}2\sqrt{\sigma_{X_1}^2\sigma_{X_2}^2}\rho_{12}\hspace{-0.5mm}+\hspace{-0.5mm}2\sqrt{\sigma_{X_1}^2\sigma_{A}^2}\rho_{1A}\hspace{-0.5mm}+\hspace{-0.5mm}2\sqrt{\sigma_{X_2}^2\sigma_{A}^2}\rho_{2A}\hspace{-0.5mm}+\hspace{-0.5mm}2\sqrt{\sigma_{X_2}^2Q}\rho_{2W}\Big{)}\Big{(}N(\rho_{1A}^2\hspace{-0.5mm}-\hspace{-0.5mm}1)\hspace{-0.5mm}-\hspace{-0.5mm}\sigma_{X_2}^2\Delta\Big{)}}{N\Big{(}(\rho_{1A}^2-1)(N+Q+\sigma_{X_2}^2\rho_{2W}^2+2\sqrt{\sigma_{X_2}^2Q}\rho_{2W})-\sigma_{X_2}^2\Delta\Big{)}}}\nonumber
\end{align}}}
\begin{align}
        \hspace{-3.6cm}&=\Omega(\sigma_{X_1}^2,\sigma_{X_2}^2,\sigma_{A}^2,Q,N,\rho_{12},\rho_{1A},\rho_{2A},\rho_{2W}).
\end{align}

To sum up, our capacity region is upper bounded by the closure of the set that contains all the rates $(R_1,R_2)$ that satisfy
\begin{eqnarray}
R_2&\leq& \Gamma(\sigma_{X_2}^2,\sigma_{A}^2,Q,N,\rho_{12},\rho_{1A},\rho_{2A},\rho_{2W})\nonumber\\
R_1+R_2&\leq& \Omega(\sigma_{X_1}^2,\sigma_{X_2}^2,\sigma_{A}^2,Q,N,\rho_{12},\rho_{1A},\rho_{2A},\rho_{2W})\label{converseRegion}.
\end{eqnarray}
for some covariance matrix
\begin{equation}\Lambda=
\left(
  \begin{array}{ccccc}
    \sigma_{X_1}^2 & \sigma_{12} & \sigma_{1A} & 0 & 0 \\
    \sigma_{12} & \sigma_{X_2}^2 & \sigma_{2A} & \sigma_{2W} & 0 \\
    \sigma_{1A} & \sigma_{2A} & \sigma_{A}^2 & 0 & 0 \\
    0 & \sigma_{2W} & 0 & Q & 0 \\
    0 & 0 & 0 & 0 & N \\
  \end{array}
\right),
\end{equation}
such that
\begin{eqnarray}
\sigma_{X_1}^2\leq P_1 \ \ \ \ \sigma_{X_2}^2\leq P_2 \ \ \ \ \sigma_{A}^2\leq P_A.\nonumber
\end{eqnarray}
and the nonnegative-definiteness condition is satisfied, i.e.,
\begin{eqnarray}
\det{(\Lambda)}=\sigma_{1A}^2 \sigma_{2W}^2 N \sigma_{X_1}^2 \sigma_{A}^2 +2\sigma_{12}\sigma_{1A}\sigma_{2A}NQ-\sigma_{2A}^2N \sigma_{X_1}^2 Q -\sigma_{12}^2 N \sigma_{A}^2 Q +N\sigma_{X_1}^2 \sigma_{X_2}^2 \sigma_{A}^2 Q\geq 0,\label{condition1}
\end{eqnarray}
or, equivalently, as a function of $\rho_{12},\rho_{1A},\rho_{2A}$ and $\rho_{2W}$
\begin{eqnarray}
1 -\rho_{12}^2-\rho_{1A}^2-\rho_{2A}^2-\rho_{2W}^2 + \rho_{1A}^2\rho_{2W}^2+2\rho_{1A}\rho_{2A}\rho_{12}\geq 0.
\end{eqnarray}
The nonnegative-definiteness condition states that all leading principal minors of the covariance matrix need to be $\geq 0$, but we can easily convince ourselves that the condition (\ref{condition1}) is sufficient. This constraint is due to the fact that any symmetric nonnegative-definite matrix is a covariance matrix. Therefore, by adding this constraint we have a one-to-one correspondence between the problem in the form of random variables and the problem in the form of  $\sigma_{X_1}^2$, $\sigma_{X_2}^2$, $\sigma_{A}^2$, $\rho_{12}$, $\rho_{1A}$, $\rho_{2A}$, $\rho_{2W}$.

\par Next, we need to show that we can replace $\sigma_{X_1}^2$, $\sigma_{X_2}^2$ and $\sigma_{A}^2$ with $P_1$, $P_2$ and $P_A$, respectively. Let us assume that the region (\ref{converseRegion}) is indeed achievable, a fact that we will prove in the following subsection. Observing the term $\Omega(\sigma_{X_1}^2,\sigma_{X_2}^2,\sigma_{A}^2,Q,N,\rho_{12},\rho_{1A},\rho_{2A},\rho_{2W})$, which is dependent on $\sigma_{X_1}^2$, we can see that replacing $\sigma_{X_1}^2$ with $P_1$ increases the region. This is due to the fact that $\Omega(\sigma_{X_1}^2,\sigma_{X_2}^2,\sigma_{A}^2,Q,N,\rho_{12},\rho_{1A},\rho_{2A},\rho_{2W})$ increases with $\sigma_{X_1}^2$.  Furthermore, the term $\Gamma(\sigma_{X_2}^2,\sigma_{A}^2,Q,N,\rho_{12},\rho_{1A},\rho_{2A},\rho_{2W})$ is unaffected by $P_1$.

\par Now, we would like to show that we can replace $\sigma_{X_2}$ with $P_2$. Assume, by contradiction, that the best coding scheme is obtained for a $\sigma_{X_2}^2\leq P_2$ . Given this scheme, we reach a contradiction by showing that the informed encoder can utilize it's unused power $\tilde{P}_2=P_2-\sigma_{X_2}^2$. Consider the following new coding scheme. The informed encoder uses the given coding scheme to send its signal with power $\sigma_{X_2}^2$, but in addition, it also uses its extra unexploited power, $\tilde{P}_2$, to send the signal of the uninformed encoder. Thus, this new scheme is equivalent to the given coding scheme, but with the addition of sending the uninformed user's signal with power $\tilde{P}_2+P_1$. However, since the region (\ref{converseRegion}) increases with $\sigma_{X_1}^2$, taking $\sigma_{X_1}^2=\tilde{P}_2+P_1$ (where the extra power comes from the informed encoder, without exceeding the power constraints) increases (\ref{converseRegion}). Therefore taking $\sigma_{X_2}^2=P_2$ further increases (\ref{converseRegion}). Hence, we can replace $\sigma_{X_2}^2$ with $P_2$. A similar argument applies to taking $\sigma_A^2=P_A$.

\par Finally, after establishing the upper bound, we give an equivalent representation for it in the next equalities, (\ref{eq1.1})-(\ref{eq1.2}). This form of the upper bound is needed for the direct part. There, we will show that we can achieve this exact form from the general region of the Action-MAC; hence, we can conclude that the bound is tight.
\vspace{-3mm}

\begin{align}
I(A;Y&|X_1)+I(X_2;Y|X_1,A,W)\nonumber\\
           &=I(A;Y|X_1)+h(X_2|X_1,A,W)-h(X_2|X_1,A,W,Y)\nonumber\\
           &=I(A;Y|X_1)+h(X_2|X_1,A,W)-h(X_2-\hat{X}_2^{\text{lin}}(X_1,A,W,X_2+Z)|X_1,A,W,X_2+Z)\nonumber\\
           &=I(A,X_1;Y)+h(X_2|X_1,A,W)-h(X_2-(\beta_1 X_1+\beta_2 A +\beta_3 W+\beta_4 (X_2+Z))),\label{eq1.1}
\end{align}
and similarly,
\begin{align}
I(A,X_1;&Y)+I(X_2;Y|X_1,A,W)\nonumber\\
           &=I(A,X_1;Y)+h(X_2|X_1,A,W)-h(X_2-\hat{X}_2^{\text{lin}}(X_1,A,W,X_2+Z))\nonumber\\
           &=I(A,X_1;Y)+h(X_2|X_1,A,W)-h(X_2-(\beta_1 X_1+\beta_2 A +\beta_3 W+\beta_4 (X_2+Z))),\label{eq1.2}
\end{align}
where $\hat{X}_2(X_1,A,W,X_2+Z)=E[X_2|X_1,A,W,X_2+Z]$ is the optimal MMSE estimator of $X_2$ given $(X_1,A,W,X_2+Z)$, where $(X_1,X_2,A,W,Z,Y)$ are jointly Gaussian, hence it is also the linear optimal estimator, i.e. $\hat{X}_2=\hat{X}_2^{\text{lin}}$ and

\begin{equation}
\hat{X}_2=\hat{X}_2^{\text{lin}}(X_1,A,W,X_2+Z)=\left(
                                        \begin{array}{cccc}
                                          \sigma_{12} & \sigma_{2A} & \sigma_{2W} & P_2 \\
                                        \end{array}
                                      \right)\cdot \left(
                                                     \begin{array}{cccc}
                                                       P_1 & \sigma_{1A} & 0 & \sigma_{12} \\
                                                       \sigma_{1A} & P_A & 0 & \sigma_{2A} \\
                                                       0 & 0 & Q& \sigma_{2W} \\
                                                       \sigma_{12} & \sigma_{2A} & \sigma_{2W} & P_2+N\\
                                                     \end{array}
                                                   \right)^{-1}\cdot \left(
                                                                  \begin{array}{c}
                                                                    X_1 \\
                                                                    A\\
                                                                    W \\
                                                                    X_2+Z \\
                                                                  \end{array}
                                                                \right).
\end{equation}

$(\beta_1,\beta_2,\beta_3,\beta_4)$ are taken to be the coefficients of the optimal linear estimator of $X_2$ given $(X_1,A,W,X_2+Z)$.

To conclude the converse, in Lemma \ref{GaussaianLemma1} we gave a general upper bound to the region (\ref{Corollary1}). Lemma \ref{GaussaianLemma2} showed that it suffices to consider only jointly Gaussian random variables. Next, we calculated the expressions for the capacity region and showed that we can replace $\sigma_{X_1}^2$, $\sigma_{X_2}^2$ and $\sigma_{A}^2$ with $P_1$, $P_2$ and $P_A$, respectively. Finally, we derived a set of equalities, (\ref{eq1.1})-(\ref{eq1.2}), that will be shown to be achievable in the direct part.

\subsection{Direct Part}\label{GaussianDirect}
In this section, we will choose specific distributions of our random variables and place them in the capacity region (\ref{Corollary1}), given in Corollary \ref{intuitiveRegion} of Theorem \ref{TheoremMain}. We will see that, by this choice, we are able to achieve the upper bound (\ref{eq1.1})-(\ref{eq1.2}) found in the converse, concluding that it is, indeed, tight. To do so, let us choose $(X_1,X_2,A,W,Y)$ to be jointly Gaussian with
\begin{eqnarray}
E[X_1^2]=P_1,\ \ \ \ \ \ E[X_2^2]=P_2, \ \ \ \ \ \ E[A^2]=P_A.
\end{eqnarray}
In addition, we take the auxiliary random variable $U$ to be $U= X_1+X_2+\tilde{\beta}_3 S$ (the notation $\tilde{\beta}_3$ will become clearer later). Note that $U$ is a function of $X_1,X_2$ and $S$ so that the Markov $U-(X_1,X_2,S,A)-Y$ holds. Now, take the expression for $R_1+R_2$ from (\ref{Corollary1}), i.e. $I(X_1,A,U;Y)-I(X_1,U;S|A)$:
\begin{align}
I(X_1,A,U;Y)-&I(X_1,U;S|A)= I(X_1,A;Y)+ I(U;Y|X_1,A)-I(U;W|X_1,A)\nonumber\\
                         &= I(X_1,A;Y)+ h(U|W,X_1,A)-h(U|Y,X_1,A)\nonumber\\
                         &\stackrel{(a)}{=} I(X_1,A;Y)+ h(X_1+X_2+\tilde{\beta}_3(A+W)|X_1,A,W)-h(X_1+X_2+\tilde{\beta}_3(A+W)|Y,X_1,A)\nonumber\\
                         &= I(X_1,A;Y)+ h(X_2|X_1,A,W)-h(X_2+\tilde{\beta}_3 W|Y,X_1,A)\nonumber\\
                         &\stackrel{(b)}{=} I(X_1,A;Y)+ h(X_2|X_1,A,W)-h(X_2+\tilde{\beta}_3 W-\beta_4 Y|Y,X_1,A)\nonumber\\
                         &\stackrel{(c)}{=} I(X_1,A;Y)+ h(X_2|X_1,A,W)-h(X_2+\tilde{\beta}_3 W-\beta_4 (X_2+W+Z)|Y,X_1,A)\nonumber\\
                         &\stackrel{(d)}{=} I(X_1,A;Y)+ h(X_2|X_1,A,W)-h(X_2-\beta_3 W-\beta_4 (X_2+Z)|Y,X_1,A)\nonumber\\
                         &=I(X_1,A;Y)+ h(X_2|X_1,A,W)-h(X_2-\beta_3 W-\beta_4 (X_2+Z)-\beta_1 X_1-\beta_2 A|Y,X_1,A)\nonumber\\
                         &\stackrel{(e)}{=} I(X_1,A;Y)+ h(X_2|X_1,A,W)-h(X_2-\beta_1 X_1-\beta_2 A-\beta_3 W-\beta_4 (X_2+Z)),
\end{align}
where\\
$(a)$ is achieved by submitting $U= X_1+X_2+\tilde{\beta}_3 S$,\\
$(b)$ is due to the fact that adding a deterministic term, in this case $-\beta_4 Y$, does not effect the entropy,\\
$(c)$ is achieved be replacing $Y$ with $X_1+X_2+A+W+Z$ (according to the channel model) where $X_1$ and $A$ are conditioned upon and,therefore, known,\\
$(d)$ is due to the fact that we take $-\beta_3=\tilde{\beta}_3-\beta_4$.
\par Finally, to show equality $(e)$, let us look at the expression $h(X_2-\beta_1 X_1-\beta_2 A-\beta_3 W-\beta_4 (X_2+Z)|Y,X_1,A)$. It is clear that
\begin{equation}
h(X_2-\beta_1 X_1-\beta_2 A-\beta_3 W-\beta_4 (X_2+Z)|Y,X_1,A)\leq h(X_2-\beta_1 X_1-\beta_2 A-\beta_3 W-\beta_4 (X_2+Z)),\nonumber
\end{equation}
 since conditioning reduces entropy. Now consider
\begin{align*}
h(X_2-\beta_1 X_1-\beta_2 A-\beta_3 W-\beta_4 (X_2+Z)|Y,X_1,A)\\
                            &\hspace{-20mm}\geq h(X_2-\beta_1 X_1-\beta_2 A-\beta_3 W-\beta_4 (X_2+Z)|Y,X_1,A,W)\\
                            &\hspace{-20mm}= h(X_2\hspace{-0.5mm}-\hspace{-0.5mm}\beta_1 X_1\hspace{-0.5mm}-\hspace{-0.5mm}\beta_2 A\hspace{-0.5mm}-\hspace{-0.5mm}\beta_3 W\hspace{-0.5mm}-\hspace{-0.5mm}\beta_4 (X_2\hspace{-0.5mm}+\hspace{-0.5mm}Z)|X_1\hspace{-0.5mm}+\hspace{-0.5mm}X_2\hspace{-0.5mm}+\hspace{-0.5mm}A\hspace{-0.5mm}+\hspace{-0.5mm}W\hspace{-0.5mm}+\hspace{-0.5mm}Z,X_1,A,W)\\
                            &\hspace{-20mm}= h(X_2-\beta_1 X_1-\beta_2 A-\beta_3 W-\beta_4 (X_2+Z)|X_2+Z,X_1,A,W)\\
                            &\hspace{-20mm}= h(X_2-\beta_1 X_1-\beta_2 A-\beta_3 W-\beta_4 (X_2+Z)),
\end{align*}
where the last step follows from the fact that we can choose $\beta_1$, $\beta_2$, $\beta_3$ and $\beta_4$ such that $(X_2-\beta_1 X_1-\beta_2 A-\beta_3 W-\beta_4 (X_2+Z))$ is the error of the optimal MMSE estimate of $X_2$ given $(X_2+Z,X_1,A,W)$ and thus independent of $(X_2+Z,X_1,A,W)$. Therefore, we have the following equality
\begin{equation}
h(X_2-\beta_1 X_1-\beta_2 A-\beta_3 W-\beta_4 (X_2+Z)|Y,X_1,A)= h(X_2-\beta_1 X_1-\beta_2 A-\beta_3 W-\beta_4 (X_2+Z)).\nonumber
\end{equation}
The direct proof for the bound on $R_2$ is obtained in a similar manner.

\par Summing up the direct part, we have shown that for this choice of $U$ and by choosing $(X_1,X_2,A,W,Y)$ to be jointly Gaussian, we achieve the upper bound (\ref{eq1.1})-(\ref{eq1.2}) from the converse section. Hence, we obtain a tight expression for the capacity region of this channel.


\subsection{Action-dependent point-to-point Channel}
The capacity for the action-dependent point-to-point channel was left open in \cite{Tsachy-Weissman}. An achievable scheme was given, but was not shown to be tight. Here, we use similar steps to the proof of the Gaussian Action-MAC to find the capacity of the point-to-point channel. Moreover, since the point-to-point channel is a special case of the Action-MAC, its capacity expression can be derived directly from the capacity region of the Gaussian Action-MAC found in the previous subsection. This is done by considering only the private message, i.e. taking $R_1=0$.

\par The outline for the direct proof of the capacity expression is given here. An alternative proof is available in Appendix \ref{AppendixPTPGauss}. In the alternative proof we show a one-to-one correspondence to the previously solved setting of the GGP MAC with only a common message \cite{Baruch-Shamai-Verd}. After establishing a correspondence, we use the result of the GGP MAC to find the capacity of the point-to-point channel. Furthermore, we show that the achievable region found in \cite{Tsachy-Weissman} is equivalent to the capacity expression and is, indeed, tight.
\par Let us recall the setting of the Gaussian action-dependent point-to-point channel \cite{Tsachy-Weissman}. Here, the channel probability is given by the equation
\begin{eqnarray}
Y_i=X_i+S_i+Z_i=X_i+A_i+W_i+Z_i,
\end{eqnarray}
where $W$ and $Z$ are defined as before, in addition to the power constraints:
\begin{eqnarray}
\frac{1}{n}\sum_{i=1}^n A_i(M)\leq P_A,\ \ \ \ \ \ \ \frac{1}{n}\sum_{i=1}^n X_{i}(M)^2\leq P_X.
\end{eqnarray}
Now, using similar steps as in the proof of Theorem \ref{GaussianTheorem} we obtain a closed capacity expression for the Gaussian action-dependent point-to-point channel given in the next corollary.
\begin{corollary}\label{GaussPTPCorollary1}
The capacity of the Gaussian action-dependent point-to-point channel is
\begin{eqnarray}
C= \frac{1}{2}\log\Big{(}\frac{(N+Q+P_X+P_A+2\sqrt{P_XP_A}\rho_{XA}+2\sqrt{P_XQ}\rho_{XW})(N+P_X(1-\rho_{2A}^2-\rho_{XW}^2))}{N(N+Q+P_X\rho_{XW}^2+2\sqrt{P_XQ}\rho_{XW}-P_X(1-\rho_{XA}^2-\rho_{XW}^2))} \Big{)},\label{GaussainPTP1}
\end{eqnarray}
for some $\rho_{XA}\in[0,1]$ and $\rho_{XW}\in[-1,0]$, such that
\begin{eqnarray}
\rho_{XA}^2+\rho_{XW}^2\leq 1,
\end{eqnarray}
where
\begin{eqnarray}
\rho_{XA}=\frac{\sigma_{XA}}{\sqrt{P_AP_X}},\ \ \rho_{XW}=\frac{\sigma_{XW}}{\sqrt{P_X\sigma_W^2}}.\nonumber
\end{eqnarray}
\end{corollary}

Let us give an outline for the proof, using similar steps as for the proof of the Gaussian Action-MAC. First, recall the capacity expression for the general action-dependent point-to-point channel:
\begin{equation}
C=I(A,U;Y)-I(U;S|A). \label{ptpCapacity}
\end{equation}
Taking similar steps to those taken in the converse of the Gaussian Action-MAC setting, we can show that (\ref{ptpCapacity}) is upper bounded by:
\begin{equation}
I(A,U;Y)-I(U;S|A)\leq I(A;Y)+I(X;Y|W,A)\label{ptpGauss1}.
\end{equation}
Second, similarly to Lemma \ref{GaussaianLemma2}, we show that (\ref{ptpGauss1}) is maximized by taking $(X,A,W,Z,Y)$ to be jointly Gaussian.
The next step is to calculate the expression
\begin{align}
I(A;Y)+I(X;Y|W,A)&=h(Y)+h(W|Y,A)-h(W)-h(Z)\nonumber\\
                 &=\frac{1}{2}\log{(\frac{\sigma_{Y}^2\sigma_{W|Y,A}^2}{QN})}.
\end{align}
We calculate the expression for fixed second moments
\begin{eqnarray}
E[X_2]=\sigma_X^2\leq P_X \ \ \ \ \ E[A^2]=\sigma_A^2\leq P_A \ \ \ \ \ E[XA]=\sigma_{XA} \ \ \ \ \ E[XW]=\sigma_{XW}
\end{eqnarray}
and show that replacing $\sigma_X^2$ and $\sigma_A^2$ with $P_X$ and $P_A$, respectively, further increases the region. Following these steps shows that (\ref{GaussainPTP1}) is an upper bound for the capacity of the Gaussian setting.

\par Third, with the next set of equalities we can show that this upper bound is achievable in the same manner as in section \ref{GaussianDirect}.
\begin{equation}
I(A;Y)+I(X;Y|S,A)=I(A;Y)+h(X|A,W)-h(X-\hat{X}^{\text{lin}}(A,W,X+Z))\label{ptpAchieve},
\end{equation}
where $\hat{X}^{\text{lin}}(A,W,X+Z)$ is the optimal linear MMSE estimator of $X$ given $(A,W,X+Z)$.
For the direct part, we can achieve (\ref{ptpAchieve}) from the general capacity expression by taking $U=X+\beta S$, where $S= A+W$. Therefore, we can conclude that the expression (\ref{GaussainPTP1}) is indeed, tight and is the capacity for the Gaussian action-dependent point-to-point channel.

\par Alternatively, we can also derive (\ref{GaussainPTP1}) in two additional manners. Firstly, since the point-to-point channel is a special case of the MAC, we can derive (\ref{GaussainPTP1}) directly from the region (\ref{RegionTheorem2}) given in Theorem \ref{GaussianTheorem}. This is done by only considering $R_2$, i.e. $R_1=0$, and taking the following variables to be $P_1 \longrightarrow  0$, $P_2\longrightarrow P_X$, $P_A \longrightarrow P_A$, $Q \longrightarrow Q$, $N \longrightarrow N$, $\rho_{12} \longrightarrow 0$, $\rho_{1A} \longrightarrow 0$, $\rho_{2A} \longrightarrow \rho_{XA}$ and $\rho_{2W} \longrightarrow \rho_{XW}$.
Secondly, we show in Appendix \ref{AppendixPTPGauss} that we can obtain (\ref{GaussainPTP1}), using the equivalent representation
\begin{eqnarray}
C= \frac{1}{2}\log\Big{(}1+ \frac{P_X(1-\rho_{XA}^2-\rho_{XW}^2)}{N}\Big{)}+ \frac{1}{2}\log\Big{(}1+ \frac{(\sqrt{P_A}+\rho_{XA}\sqrt{P_X})^2}{P_X(1-\rho_{XA}^2-\rho_{XW}^2)+(\sigma_W +\rho_{XW}\sqrt{P_X})^2+N}\Big{)}.\label{GaussainPTP22}
\end{eqnarray}
The representation (\ref{GaussainPTP22}) is simply obtained by reorganizing (\ref{GaussainPTP1}). We prove the capacity of the action-dependent point-to-point channel is equal to the equivalent representation (\ref{GaussainPTP22}) by using a one-to-one correspondence to the GGP MAC \cite{Baruch-Shamai-Verd} with only a common message.


\section{Duality Channel-Source Coding with Action}\label{DualSection}


\par In this section we recognize and discuss the duality between the MAC with action-dependent state information at one encoder and the ``Successive Refinement with Actions'' rate distortion setting. The information-theoretic duality between channel coding and source coding was first recognized by Shannon in his milestone paper \cite{Shannon60}. Chiang and Cover \cite{ChiangCover01}, as well as Pradhan, Chou and Ramchandran \cite{Pradhan_duality03}, further expanded the discussion of duality to the basic settings with side information.  In addition, duality for source and channel coding with action-dependent side information was noted by Kittichokechai, Tobias, Oechterin and Skoglund \cite{KittichokechaiTobias}.
\par Source coding with side information was presented by Wyner and Ziv in their landmark paper \cite{Wyner_ziv76_side_info_decoder}. They discussed the case where side information is available noncausally at the decoder and not available at the encoder. This case is dual to the Gel'fand-Pinsker \cite{Gelfand-Pinsker} channel coding problem, as shown in \cite{ChiangCover01} and \cite{Pradhan_duality03}. Operational duality between these two setting was shown by Gupta and Verdu \cite{GuptaVerdue_Duality10}. Rate distortion problems for two decoders, where only one of them has access to correlated side information were considered by Kaspi \cite{Kaspi}, as well as by Heegrad and Berger \cite{HeegardBerger}. This model was extended by Steinberg and Merhav in \cite{Steinberg_merhav04_sucessuve_refienment_wyner_ziv} to successive refinement with side information. A successive refinement model with conditionally less noisy side information was studied by Timo, Oechtering and Wigger in \cite{TimoOechteringWigger12}.

\par The novel idea of source coding where the system can take actions that affect the availability, quality or nature of the side information was introduced in Weissman et al. in \cite{VendingMachine}. They presented the rate distortion setting with a side information vender at the decoder. This setting was shown to be dual to Weisman's action-dependent point-to-point channel in \cite{KittichokechaiTobias}. The ``Successive Refinement with Actions'' rate distortion setting, introduced in \cite{Khiang-Asnani-Weissman}, is an extension of these previously discussed rate distortion settings to the case where we have two decoders, where only one has access to vender side information. Additionally, the encoder transmits the source sequence to the decoders with a common rate and a private rate. This setting is illustrated in Fig. \ref{RateDistortionFig}.

\par Here, we characterize the dual relationship between the ``Successive Refinement with Actions'' rate distortion setting and the Action-MAC channel coding setting. We start by revising the ``Successive Refinement with Actions'' rate distortion model as well as the theorem for the rate region $\mathcal{R}(D_1,D_2)$. Next, we state the equivalence relationships of these channel coding and rate distortion problems and give a set of duality principles that show the connections between them. We end this section with a detailed discussion on the duality between the two settings; this is followed by additional duality relationships that are deduced as special cases. For instance, we recognize the duality between MAC with state information at one encoder \cite{Baruch-Shamai-Verd} and a special case of the ``Successive Refinement for the Wyner-Ziv problem'' \cite{Steinberg_merhav04_sucessuve_refienment_wyner_ziv}.

 \begin{figure}[h!]
    \begin{center}
        \begin{psfrags}
            \psfragscanon
            \psfrag{A}[][][0.8]{\ \ \ \ \ \ \ \ \ \ $X^n$}
            \psfrag{B}[][][0.8]{\ \ \ \ \ \ \ \ \ \ \ \ \ \ \ \ \ \ \ \ $\hat{X}_1^n$}
            \psfrag{C}[][][0.8]{\ \ \ \ \ \ \ \ \ \ \ \ \ \ \ \ \ \ \ \ $\hat{X}_2^n$}
            \psfrag{D}[][][0.8]{\ \ \ \ \ \ \ \ \ \ \ \ \ \ \ \ \ \ \ \ \ \ \ \ \ \ \ $T_1(X^n)\in \{1,2,...,2^{nR_1}\}$}
            \psfrag{E}[][][0.8]{\ \ \ \ \ \ \ \ \ \ \ \ \ \ \ \ \ \ \ \ \ \ \ \ \ \ \ $T_2(X^n)\in \{1,2,...,2^{nR_2}\}$}
            \psfrag{F}[][][0.8]{\ \ \ \ \ \ \ \ \ \ \ \ \ \ \ Encoder}
            \psfrag{G}[][][0.8]{\ \ \ \ \ \ \ \ \ \ \ \ \ \ \ \ \ Uninformed}
            \psfrag{H}[][][0.8]{\ \ \ \ \ \ \ \ \ \ \ \ \ \ \ \ \ Informed}
            \psfrag{I}[][][0.8]{\ \ \ \ \ \ \ \ \ \ \ \ \ \ \ \ \ Decoder}
            \psfrag{J}[][][0.8]{\ \ \ \ \ \ \ \ \ \ \ \ \ \ \ \ \ Decoder}
            \psfrag{K}[][][0.8]{\ \ \ \ \ \ \ \ \ \ \ \ \ \ \ \ Vender}
            \psfrag{L}[][][0.8]{\ \ \ \ \ \ \ \ \ \ \ \ \ \ \ \ \ $p(s|a,x)$}
            \psfrag{M}[][][0.8]{\ \ \ \ \ \ \ \ \ $A^n(T_1,T_2)$}
            \psfrag{N}[][][0.8]{$S^n$}
            \centerline{\includegraphics[scale = .6]{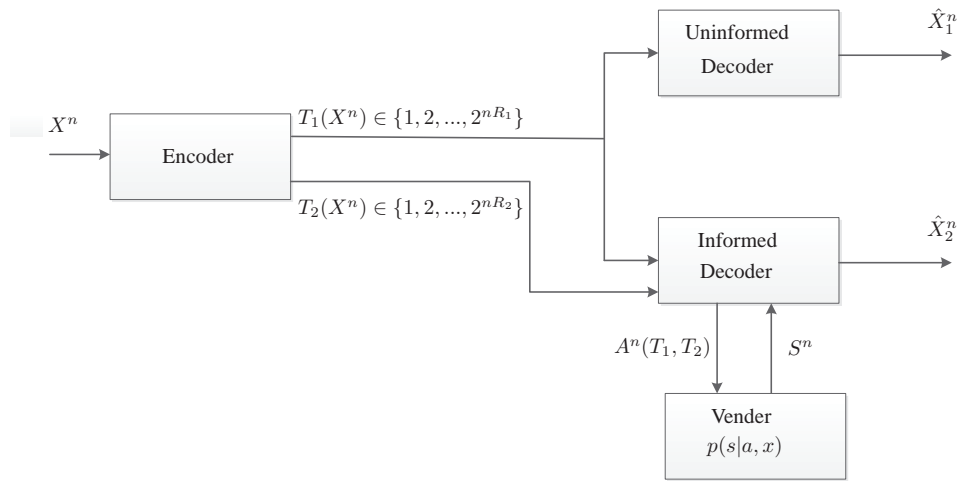}}
            \caption{The ``Successive Refinement with Actions'' rate distortion setting, presented in \cite{Khiang-Asnani-Weissman}} \label{RateDistortionFig}
            \psfragscanoff
        \end{psfrags}
     \end{center}
 \end{figure}

\subsection{The ``Successive Refinement with Actions'' rate distortion setting}\label{RateDistortionSetting}
Let us revise the formal definition of the dual rate distortion problem. We consider a source sequence $X_1,X_2,...,X_n$ of i.i.d. random variables where each variable is drawn $\sim p(x)$ from a source alphabet $\mathcal{X}$. Let $\hat{\mathcal{X}}_1$ and $\hat{\mathcal{X}}_2$ denote the reconstruction alphabets and consider two, single letter distortion measures
\begin{eqnarray}
d_j:\mathcal{X}\times\hat{\mathcal{X}}_j\rightarrow\mathbb{R},\ \ \ j=1,2.
\end{eqnarray}
The distortion between $n$-sequences in $\mathcal{X}\times\hat{\mathcal{X}}_j$ is defined in the usual way:
\begin{eqnarray}
d(x^n,\hat{x}_j^n)=\frac{1}{n}\sum_{i=1}^n d(x_i,\hat{x}_{j,i}), \ \ \ \text{for } j=1,2.
\end{eqnarray}

\begin{definition}
A $((2^{nR_1},2^{nR_2}),n)$ rate distortion code consists of an encoder and two sets of decoders, one which receives the message at rate $R_1$ and one which receives the message at rates $R_1$,$R_2$. The code is defined by encoding functions
\begin{eqnarray}
T_{1}&:&\mathcal{X}^n\rightarrow \{1,2,...,2^{nR_1}\}\nonumber\\
T_{2}&:&\mathcal{X}^n\rightarrow \{1,2,...,2^{nR_2}\},
\end{eqnarray}
an action strategy
\begin{eqnarray}
f&:&\{1,2,...,2^{nR_1}\}\times\{1,2,...,2^{nR_2}\}\rightarrow \mathcal{A}^n,
\end{eqnarray}
and decoding functions
\begin{eqnarray}
g_{1,n}&:&\{1,2,...,2^{nR_1}\}\rightarrow\hat{\mathcal{X}}_1^n \nonumber\\
g_{2,n}&:&\mathcal{S}^n\times\{1,2,...,2^{nR_1}\}\times\{1,2,...,2^{nR_2}\}\rightarrow\hat{\mathcal{X}}_2^n
\end{eqnarray}
\end{definition}
\begin{definition}
A quintuple $(R_1,R_2,D_1,D_2,C)$ is said to be achievable if, for all $\epsilon >0$ and for a sufficiently large $n$, there exists a $((2^{nR_1},2^{nR_2}),n)$ code satisfying
\begin{eqnarray}
E\Big{[}\sum_{i=1}^n d(X_i,\hat{X}_{j,i})\Big{]}\leq n(D+\epsilon)\ \ \ \text{for}\ \ \  j=1,2
\end{eqnarray}
and
\begin{eqnarray}
E\Big{[}\sum_{i=1}^n \Lambda(A_i)\Big{]}\leq n(C+\epsilon),
\end{eqnarray}
where $\Lambda$ is the cost function.
\end{definition}
\begin{definition}
 The rate distortion region $\mathcal{R}(D_1D_2)$ is defined as the closure of the set for all achievable
rate distortion quintuples $(R_1,R_2,D_1,D_2,C)$.
\end{definition}

The rate region $\mathcal{R}(D_1D_2)$ for the ``Successive Refinement with Actions'' setting, illustrated in Fig. \ref{RateDistortionFig}, was proven in \cite{Khiang-Asnani-Weissman} and is the closure of the set of all the rate tuples $(R_1,R_2)$ such that,
\begin{eqnarray}
R_1&\geq& I(X;\hat{X}_1)\nonumber\\
R_1+R_2&\geq& I(X;\hat{X}_1)+ I(A;X|\hat{X_1})+I(X;U|S,A,\hat{X}_1), \label{DistortionRegion}
\end{eqnarray}
for some joint distribution
\begin{eqnarray}
P(x,a,u,s,\hat{x}_1)=P(x)P(a,u,\hat{x}_1|x)P(s|x,a),
\end{eqnarray}
such that
\begin{eqnarray}
E\Big{[}d_1(X,\hat{X}_{1})\Big{]}&\leq& D_1,\\
E\Big{[}d_2(X,\hat{X}_{2}(S,U))\Big{]}&\leq& D_2,\\
E\Big{[}\Lambda(A)\Big{]}&\leq& C.
\end{eqnarray}


\subsection{Duality results between the Action-MAC and the ``Successive Refinement with Actions''}

\par Using a set of simple duality transformation principles we obtain a precise characterization of the functional duality between the Action-MAC and the ``Successive Refinement with Actions'' rate distortion setting. The duality transformation principles between the two settings are given in Table \ref{DualityTable}. In other words, for a given channel coding problem we obtain a rate distortion problem and vice versa. Here, the roles of encoder and decoder are functionally interchangeable and the input-output joint distribution is equivalent with some renaming of variables.

\begin{center}
\begin{table}[h!]
\caption{Duality transformation principles between the ``Successive Refinement with Actions'' rate distortion setting and the Action-MAC}\label{DualityTable}
\hspace{13mm}\begin{tabular}{|c|c|}
\hline
\textbf{``Successive Refinement with Actions''} & \textbf{Action-MAC}\\ \hline\hline
 $R\geq R(D_1,D_2)$ & $R\leq C$  \\ \hline
Decoder inputs / Encoder outputs: & Encoder inputs / Decoder outputs:\\
$T_1\in\{1,2,...,2^{nR_1}\}$ & $M_1\in\{1,2,...,2^{nR_1}\}$\\
$T_2\in\{1,2,...,2^{nR_2}\}$ & $M_2\in\{1,2,...,2^{nR_2}\}$\\ \hline
Decoder outputs / Source reconstruction:  & Encoder outputs / Channel input:\\
$\hat{X}_1^n$, $\hat{X}_2^n$ & $X_1^n$, $X_2^n$\\ \hline
Encoder input / Source:  & Decoder input / Channel output:\\
$X^n$ & $Y^n$ \\ \hline
Decoder functions: & Encoder functions:\\
$g_1:\{1,2,...,2^{nR_1}\}\rightarrow \mathcal{\hat{X}}_1^n$ & $f_1:\{1,2,...,2^{nR_1}\}\rightarrow \mathcal{X}_1^n$\\
$g_2:\{1,2,...,2^{nR_1}\}\times\{1,2,...,2^{nR_2}\}\times \mathcal{S}^n\rightarrow \mathcal{\hat{X}}_2^n$ &
$f_2:\{1,2,...,2^{nR_1}\}\times\{1,2,...,2^{nR_2}\}\times \mathcal{S}^n\rightarrow \mathcal{X}_2^n $\\ \hline
Action strategy: & Action encoder:\\
$f_A:\mathcal{T}_1\times\mathcal{T}_2\rightarrow \mathcal{A}^n$ & $f_A:\mathcal{M}_1\times\mathcal{M}_2\rightarrow \mathcal{A}^n$\\ \hline
$U$ auxiliary random variable & $U$ auxiliary random variable\\ \hline
$S$ side information & $S$ state information\\ \hline
Joint distribution $P(x,a,s,\hat{x}_1,u,\hat{x}_2)$ & Joint distribution: $P(a,s,x_1,u,x_2,y)$\\ \hline
Markov $S-(A,X)-U,\hat{X}_1$ & Markov $S-A-X_1$\\
\hline
\end{tabular}
\end{table}
\end{center}

\begin{center}
\begin{table}[h!]
\caption{Corner points $(R_1,R_2)$ of the Action-MAC setting and its dual rate distortion setting}\label{CornerPointsTable2}
\begin{tabular}{|c||c c|}
\hline
& \textbf{$R_1$} & \textbf{$R_2$}\\ \hline\hline
\textbf{Rate} & $I(X;\hat{X}_1)+ I(A;X|\hat{X_1})+I(X;U|A,\hat{X}_1)-I(S;U|A,\hat{X}_1)$ & $0$ \\
\textbf{Distortion} & $I(\hat{X}_1;X)$ & $I(A;X|\hat{X}_1) + I(U;X|\hat{X}_1,A) - I(U;S|\hat{X}_1,A)$\\ \hline
\textbf{Action} & $I(X_1;Y)+I(A;Y|X_1)+ I(Y;U|A,X_1) - I(S;U|A,X_1)$ & $0$\\
\textbf{MAC} & $I(X_1;Y)$ & $I(A;Y|X_1) + I(U;Y|X_1,A) - I(U;S|X_1,A)$\\
\hline
\end{tabular}
\end{table}
\end{center}

 \begin{figure}[h!]
    \begin{center}
        \begin{psfrags}
            \psfragscanon
            \psfrag{A}[][][0.6]{\ \ \ \ \ \ \ \ \ \ \ \ \ \ \ \ \ \  Encoder}
            \psfrag{B}[][][0.6]{\ \ \ \ \ \ \ \ \ \ \ \ \ \ \ \ \ \  Decoder}
            \psfrag{C}[][][0.6]{\ \ \ \ \ \ \ \ \ \ \ \ \ \ \ \ \ \ \ \   Uninformed}
            \psfrag{D}[][][0.6]{\ \ \ \ \ \ \ \ \ \ \ \ \ \  \ \   Informed}
            \psfrag{E}[][][0.6]{\ \ \ \ \ \ \ \ \ \ \ \ \ \ \ \   MAC}
            \psfrag{F}[][][0.6]{\ \ \ \ \ \ \ \ \ \ \ \ \ \ \ $p(y|x_1,x_2,s)$}
            \psfrag{G}[][][0.6]{\ \ \ \ $Y^n$}
            \psfrag{H}[][][0.6]{\ \ \ \ \ $X_1^n(M_1)$}
            \psfrag{I}[][][0.6]{\ \ \ \ \ \ \ \ \ \ \ $X_2^n(M_1,M_2,S^n)$}
            \psfrag{J}[][][0.6]{$M_1$}
            \psfrag{K}[][][0.6]{$M_2$}
            \psfrag{L}[][][0.6]{\ \ \ \ \ $(\hat{M}_1,\hat{M}_2)$}
            \psfrag{M}[][][0.6]{$A^n(M_1,M_2)$}
            \psfrag{N}[][][0.6]{$S^n$}
            \psfrag{O}[][][0.6]{\ \ \ \ \ \  $p(s|a)$}
            \psfrag{P}[][][0.6]{\ \ \  $X^n$}
            \psfrag{Q}[][][0.6]{\ \ \ \ \ \ \ \ $T_1$}
            \psfrag{R}[][][0.6]{\ \ \ \ \ \ \ \ $T_2$}
            \psfrag{S}[][][0.6]{$A^n(T_1,T_2)$}
            \psfrag{T}[][][0.6]{$S^n$}
            \psfrag{U}[][][0.6]{\ \ \ \ \ \ \ \ \ \ \ \ \ \   Vender}
            \psfrag{V}[][][0.6]{\ \ \ \ \ \ \ \ \ \ \ \ \ \ \ \ \ \ $p(s|a,x)$}
            \psfrag{W}[][][0.6]{$\hat{X}_1$}
            \psfrag{X}[][][0.6]{$\hat{X}_2$}
            \psfrag{Y}[][][0.6]{\ \ \ \ \ \ \ \ \ \ \ \ \ \ \ \ \ \    Uninformed}
            \psfrag{Z}[][][0.6]{\ \ \ \ \ \ \ \ \ \ \ \ \ \  \ \ \ \   Informed}
            \psfrag{a}[][][0.6]{\ \ \ \  $R_2$}
            \psfrag{b}[][][0.6]{\ \ \ \ \ \ \ \  $R_1$}
            \psfrag{c}[][][0.6]{\ \ \ \ \ \ \ \ \ \ \ \ \ \ \ \ \ \ \ \ \ \ \ \ \ \ \ \ \ \ \ \ \ \ \ \ \ \ \ \ \ \ \ \ \ \ \ \ \ \ $I(A;X|\hat{X}_1) + I(U;X|\hat{X}_1,A) - I(U;S|\hat{X}_1,A)$}
            \psfrag{d}[][][0.6]{\ \ \ \ \ \ \ \ \ \ \ \ \ \ \ \ \ \ \ \ \ \ \ \ \ \ \ \ \ \ \ \ \ \ \ \ \ \ \ \ \ \ \ \ \ \ \ \ \ \  $I(A;Y|X_1) + I(U;Y|X_1,A) - I(U;S|X_1,A)$}
            \psfrag{e}[][][0.6]{\ \ \ \ \ \ \ \ \ \ $I(\hat{X}_1;X)$}
            \psfrag{f}[][][0.6]{\ \ \ \ \ \ \ \ \ \ $I(X_1;Y)$}
            \psfrag{g}[][][0.6]{\ \ \ \ \ \ \ \ \ \ \ \ \ \ \ \ \ \ \ \ \ \ \ \ \ \ \ \ \ \ \ \ \ \ \ \ \ \ \ \ \ \ \ \ \ \ \ \ \ \ \ \ \ \ \ \ \ \ \ \ $I(X;\hat{X}_1)+ I(A;X|\hat{X_1})+I(X;U|A,\hat{X}_1)-I(S;U|A,\hat{X}_1)$}
            \psfrag{h}[][][0.6]{\ \ \ \ \ \ \ \ \ \ \ \ \ \ \ \ \ \ \ \ \ \ \ \ \ \ \ \ \ \ \ \ \ \ \ \ \ \ \ \ \ \ \ \ \ \ \ \ \ \ \ \ \ \ \ \ \ \ \ \ $I(X_1;Y)+I(A;Y|X_1)+ I(Y;U|A,X_1) - I(S;U|A,X_1)$}
            \psfrag{i}[][][0.8]{\ \ \ \ \ \ \ \ \ \ \ \ \ \ \ \ \ \ \ \ \ \ \ \ \ \ \ \ \ \ \ \ \ \ \ \ \ \ \ \ \ \ \ \ \ \ \ \ \ \ \ \ \ \ \ \ \ \  ``Successive Refinement with Actions'' }
            \psfrag{j}[][][0.8]{\ \ \ \ \ \ \ \ \ \ \ \ \ \ \ \ \ \ \ \ \ \ \ \ \ \ \ \ \ \ \ \ \ \ \ \ \ \ \ \ \ \ \ \ \ \ \ \ \ \ \ \ \ \ \ \ \ \  }
            \psfrag{k}[][][0.8]{\ \ \ \ \ \ \ \ \ \ \ \ \ \ \ \ \ \ \ \ \ \ \ \ \ \ \ \ \ \ \ \ \ \ \ MAC with Action-Dependent State Information}
            \psfrag{l}[][][0.8]{\ \ \ \ \ \ \ \ \ \ \ \ \ \ \ \ \ \ \ \ \ \ \ \ \ \ \ \ \ \ \ \ \ \ \ at One Encoder}
            \centerline{\includegraphics[scale = .8]{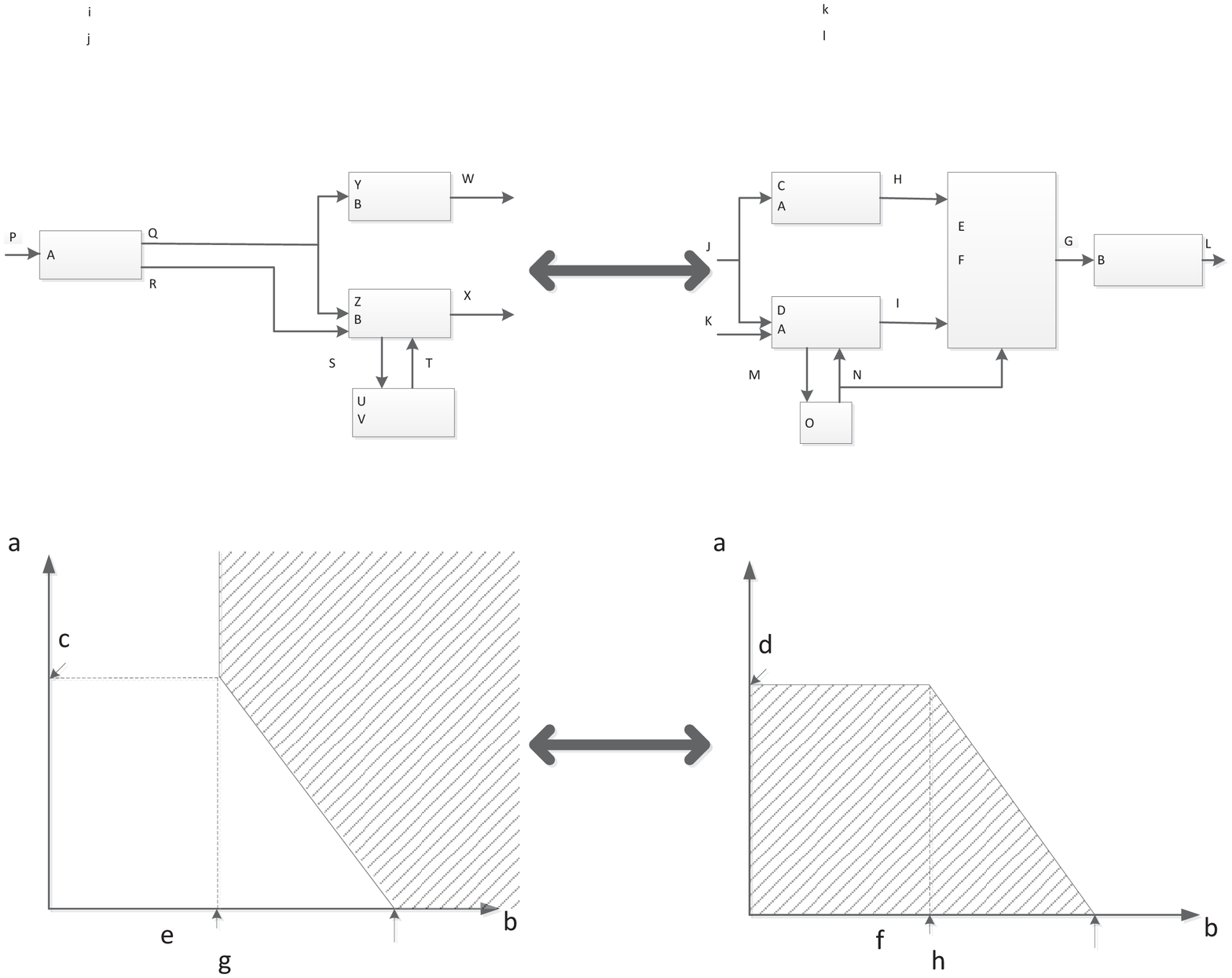}}
            \caption{A summary of the duality results. In the figure we see that if we take out the channel block from the Action-MAC model, we have an exact mirror reflection of the two settings. In this mirror reflection the roles of encoder and decoder are functionally interchangeable. The corner points are of the same form, with an exchange of variables $X\leftrightarrow Y$ and $\hat{X}_1\leftrightarrow X_1$. The bound on $R_1+R_2$ is also of the same form, with the same exchange of variables and with the additional exchange of $\geq\leftrightarrow\leq$.} \label{CornerFig2}
            \psfragscanoff
        \end{psfrags}
     \end{center}
 \end{figure}

\par To further investigate the duality between the settings, we examine their regions' corner points. Observing the resemblance between the regions' corner points helps us gain better intuition regarding their correspondence. The resemblance is not clear from simply looking at the expressions of both regions. In order to see the similarities, we need the following lemma.
\begin{lemma}\label{RateLemma}
Let $S-(A,X)-(U,\hat{X}_1)$ be a Markov chain. Then the following equality holds:
\begin{eqnarray}
I(X;\hat{X}_1)+ I(A;X|\hat{X_1})+I(X;U|S,A,\hat{X}_1)= I(X;\hat{X}_1)+ I(A;X|\hat{X_1})+I(X;U|A,\hat{X}_1)-I(S;U|A,\hat{X}_1)
\end{eqnarray}
\end{lemma}
\begin{proof}
Consider the term
\begin{align}
I(X;U|S,A,\hat{X}_1)&= I(S,X;U|S,A,\hat{X}_1)-I(S;U|A,\hat{X}_1)\nonumber\\
                    &= I(X;U|S,A,\hat{X}_1)+I(S;U|X,A,\hat{X}_1)-I(S;U|A,\hat{X}_1)\nonumber\\
                     &{=} I(X;U|A,\hat{X}_1)-I(S;U|A,\hat{X}_1)\nonumber
\end{align}
where the last equality is due to the Markov property.
\end{proof}
Using the last lemma, we can write the rate region for our rate distortion setting as
\begin{eqnarray}
R_1&\geq& I(X;\hat{X}_1)\nonumber\\
R_1+R_2&\geq& I(X;\hat{X}_1)+ I(A;X|\hat{X_1})+I(X;U|A,\hat{X}_1)-I(S;U|A,\hat{X}_1).\label{DualDis} \label{DistortionRegion2}
\end{eqnarray}
Let us recall the capacity region for our Action-MAC setting (\ref{DualCap}) given in Corollary \ref{intuitiveRegion}
\begin{eqnarray}
R_2 &\leq& I(A;Y|X_1) + I(Y;U|A,X_1) - I(S;U|A,X_1)\nonumber\\
R_1 + R_2 &\leq& I(X_1;Y)+I(A;Y|X_1)+ I(Y;U|A,X_1) - I(S;U|A,X_1).\nonumber
\end{eqnarray}
\par Looking at these two regions we notice equivalent relationships between the Action-MAC capacity region and the rate distortion rate region. Using the duality transformation principles presented in Table \ref{DualityTable}, we see an evident connection between the formulas bounding the rate sum $R_1+R_2$. This is done by consistently renaming the variables in (\ref{DistortionRegion2}), $X\leftrightarrow Y$ and $\hat{X}_1\leftrightarrow X_1$. Furthermore, by looking at these expressions for the regions of the two settings, we observe the duality of the regions' corner points. The corner points are displayed in Table \ref{CornerPointsTable2}.
\par A graphical summary of the duality results we have obtained is displayed in Fig. \ref{CornerFig2}.


\subsection{Remarks and special cases}
Throughout the paper we have seen that our Action-MAC setting generalizes two settings previously solved: the action-dependent point-to-point channel \cite{Tsachy-Weissman} and the MAC with state information available at one encoder setting (known as the GGP MAC) \cite{Baruch-Shamai-Verd}. Similarly to the Action-MAC, its dual setting, the ``Successive Refinement with Actions'', also generalizes previously solved rate distortion settings. These setting are the ``rate distortion when side information my be absent'' model, considered by Heegard and Berger \cite{HeegardBerger} as well as by Kaspi \cite{Kaspi}, a special case of the ``Successive Refinement Wyner-Ziv'' problem  \cite{Steinberg_merhav04_sucessuve_refienment_wyner_ziv} and source coding with side information ``vending machine'' \cite{VendingMachine}. In this subsection, we recognize the duality between the special cases.

\subsubsection{Duality between the Heegard-Berger/Kaspi rate distortion setting to a MAC with common message and state information}

 \begin{figure}[h!]
    \begin{center}
        \begin{psfrags}
            \psfragscanon
            \psfrag{A}[][][0.7]{\ \ \ \ \ \ \ \ \ \ \ \ \ \   Encoder}
            \psfrag{B}[][][0.7]{\ \ \ \ \ \ \ \ \ \ \ \ \ \  Decoder}
            \psfrag{C}[][][0.7]{\ \ \ \ \ \ \ \ \ \ \ \ \  Uninformed}
            \psfrag{D}[][][0.7]{\ \ \ \ \ \ \ \ \ \ \ \ \  Informed}
            \psfrag{E}[][][0.7]{\ \ \ \ \ \ \ \ \ \ \ \   MAC}
            \psfrag{F}[][][0.7]{\ \ \ \ \ \ \ \ \ \ \  \ \ $p(y|x_1,x_2,s)$}
            \psfrag{G}[][][0.7]{\ \ \ \ $Y^n$}
            \psfrag{H}[][][0.7]{\ \ \ \ \ $X_1^n(M)$}
            \psfrag{I}[][][0.7]{\ \ \ \ \ \ \ \ \ \ \ {$X_2^n(M,S^n)$}}
            \psfrag{J}[][][0.7]{$M$}
            \psfrag{K}[][][0.7]{$M_2$}
            \psfrag{L}[][][0.7]{\ \ \ \ \ $\hat{M}$}
            \psfrag{M}[][][0.7]{$A^n(M_1,M_2)$}
            \psfrag{N}[][][0.7]{$S^n$}
            \psfrag{O}[][][0.7]{\ \ \ \ \ \  $p(s|a)$}
            \psfrag{P}[][][0.7]{\ \ \  $X^n$}
            \psfrag{Q}[][][0.7]{\ \ \ \ \ \ \ \ $T$}
            \psfrag{R}[][][0.7]{\ \ \ \ \ \ \ \ $T_2$}
            \psfrag{S}[][][0.7]{$A^n(T_1,T_2)$}
            \psfrag{T}[][][0.7]{$S^n$}
            \psfrag{U}[][][0.7]{\ \ \ \ \ \ \ \ \ \ \ \ \ \   Vender}
            \psfrag{V}[][][0.7]{\ \ \ \ \ \ \ \ \ \ \ \ \ \ \ \ \ \ $p(s|a,x)$}
            \psfrag{W}[][][0.7]{$\hat{X}_1$}
            \psfrag{X}[][][0.7]{$\hat{X}_2$}
            \psfrag{Y}[][][0.7]{\ \ \ \ \ \ \ \ \ \ \ \ \ \  Uninformed}
            \psfrag{Z}[][][0.7]{\ \ \ \ \ \ \ \ \ \ \ \ \ \  Informed}
            \centerline{\includegraphics[scale = .8]{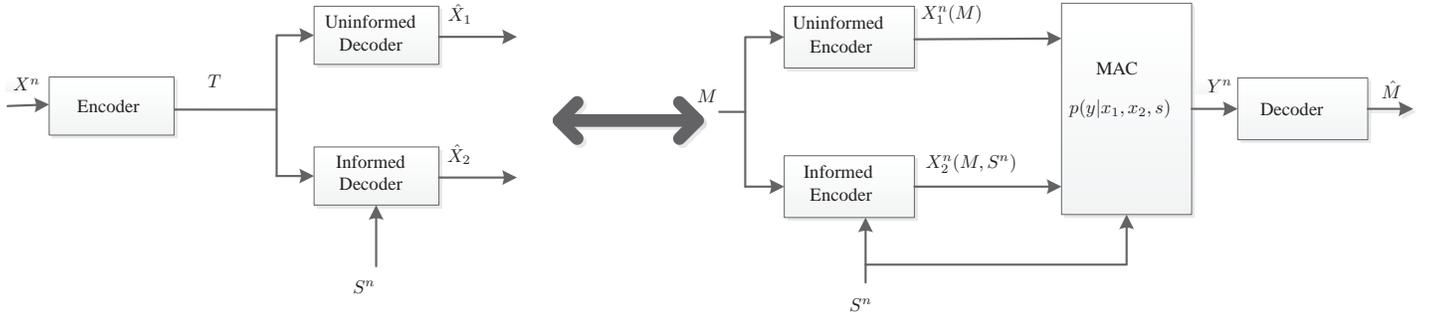}}
            \caption{The Heegard-Berger/Kaspi rate distortion setting (on the left) compared to the common message MAC with state information(on the right). If we ignore the channel block in the MAC, we have an exact mirror reflection between the settings. In this mirror reflection, the roles of the encoders and decoders are substitutable.} \label{CommonVsKaspi}
            \psfragscanoff
        \end{psfrags}
     \end{center}
 \end{figure}

  \begin{center}
\begin{table}[h!]
\caption{Duality results of the Heegard-Berger/Kaspi rate distortion setting vs. MAC with common message and state information at one encoder}\label{actionVSvending}\hspace{15mm}
\begin{tabular}{|c||c|}
\hline
\textbf{Heegard-Berger/Kaspi} & \textbf{MAC with common message}\\
\textbf{rate distortion setting} & \textbf{and state information at one encoder}\\   \hline\hline
 $R\geq R(D_1,D_2)$ & $R\leq C$  \\ \hline
 Rate distortion function & Capacity expression\\
$R(D)=I(\hat{X}_1;X)+I(U;X|\hat{X}_1)-I(U;S|\hat{X}_1)$ & $C=I(X_1;Y)+I(U;Y|X_1)-I(U;S|X_1)$\\ \hline
Decoder inputs / Encoder outputs: & Encoder inputs / Decoder outputs:\\
$T\in\{1,2,...,2^{nR_1}\}$ & $M\in\{1,2,...,2^{nR_1}\}$\\  \hline
Decoder outputs / Source reconstruction:  & Encoder outputs / Channel input:\\
$\hat{X}_1^n$, $\hat{X}_2^n$ & $X_1^n$, $X_2^n$\\ \hline
Encoder input / Source:  & Decoder input / Channel output:\\
$X^n$ & $Y^n$ \\ \hline
Decoder functions: & Encoder functions:\\
$g_1:\{1,2,...,2^{nR_1}\}\rightarrow \mathcal{\hat{X}}_1^n$ & $f_1:\{1,2,...,2^{nR_1}\}\rightarrow \mathcal{X}_1^n$\\
$g_2:\{1,2,...,2^{nR_1}\}\times \mathcal{S}^n\rightarrow \mathcal{\hat{X}}_2^n$ &
$f_2:\{1,2,...,2^{nR_1}\}\times \mathcal{S}^n\rightarrow \mathcal{X}_2^n $\\ \hline
$U$ auxiliary random variable & $U$ auxiliary random variable\\ \hline
$S$ side information & $S$ state information\\ \hline
\hline
\end{tabular}
\end{table}
\end{center}

The first special case duality we recognize is between the well known Heegard-Berger/Kaspi rate distortion setting and the MAC with common message and state information at one encoder, solved in \cite{Baruch-Shamai-Verd}. We establish this duality by using similar duality principles as those given in Table \ref{DualityTable}; i.e. we see an interchangeable relationship between the encoders and decoders in both settings, as well as the variable transformation $X\leftrightarrow Y$ and $\hat{X}_1\leftrightarrow X_1$. Given these transformations, we look at the rate distortion function for the Heegard-Berger/Kaspi model and the capacity expression for the MAC. For the rate distortion function, we use the representation given in \cite{Kaspi}, and for the capacity expression we us the common message capacity for the GGP-MAC given in \cite{Baruch-Shamai-Verd}.

\subsubsection{Duality between the GGP MAC and the Stienberg-Merhav rate distortion setting}
 \begin{figure}[h!]
    \begin{center}
        \begin{psfrags}
            \psfragscanon
            \psfrag{A}[][][0.7]{\ \ \ \ \ \ \ \ \ \ \ \ \ \   Encoder}
            \psfrag{B}[][][0.7]{\ \ \ \ \ \ \ \ \ \ \ \ \ \  Decoder}
            \psfrag{C}[][][0.7]{\ \ \ \ \ \ \ \ \ \ \ \ \  Uninformed}
            \psfrag{D}[][][0.7]{\ \ \ \ \ \ \ \ \ \ \ \ \  Informed}
            \psfrag{E}[][][0.7]{\ \ \ \ \ \ \ \ \ \ \ \   MAC}
            \psfrag{F}[][][0.7]{\ \ \ \ \ \ \ \ \ \ \  \ \ $p(y|x_1,x_2,s)$}
            \psfrag{G}[][][0.7]{\ \ \ \ $Y^n$}
            \psfrag{H}[][][0.7]{\ \ \ \ \ $X_1^n(M_1)$}
            \psfrag{I}[][][0.7]{\ \ \ \ \ \ \ \ \ \ \ \ \ \ \ \ \ \ \footnotesize{$X_2^n(M_1,M_2,S^n)$}}
            \psfrag{J}[][][0.7]{$M_1$}
            \psfrag{K}[][][0.7]{$M_2$}
            \psfrag{L}[][][0.7]{\ \ \ \ \ \ \ \ \ \ {$(\hat{M}_1,\hat{M}_2)$}}
            \psfrag{M}[][][0.7]{$A^n(M_1,M_2)$}
            \psfrag{N}[][][0.7]{$S^n$}
            \psfrag{O}[][][0.7]{\ \ \ \ \ \  $p(s|a)$}
            \psfrag{P}[][][0.7]{\ \ \  $X^n$}
            \psfrag{Q}[][][0.7]{\ \ \ \ \ \ \ \ $T_1$}
            \psfrag{R}[][][0.7]{\ \ \ \ \ \ \ \ $T_2$}
            \psfrag{S}[][][0.7]{$A^n(T_1,T_2)$}
            \psfrag{T}[][][0.7]{$S^n$}
            \psfrag{U}[][][0.7]{\ \ \ \ \ \ \ \ \ \ \ \ \ \   Vender}
            \psfrag{V}[][][0.7]{\ \ \ \ \ \ \ \ \ \ \ \ \ \ \ \ \ \ $p(s|a,x)$}
            \psfrag{W}[][][0.7]{$\hat{X}_1$}
            \psfrag{X}[][][0.7]{$\hat{X}_2$}
            \psfrag{Y}[][][0.7]{\ \ \ \ \ \ \ \ \ \ \ \ \ \  Uninformed}
            \psfrag{Z}[][][0.7]{\ \ \ \ \ \ \ \ \ \ \ \ \ \  Informed}
            \centerline{\includegraphics[scale = .8]{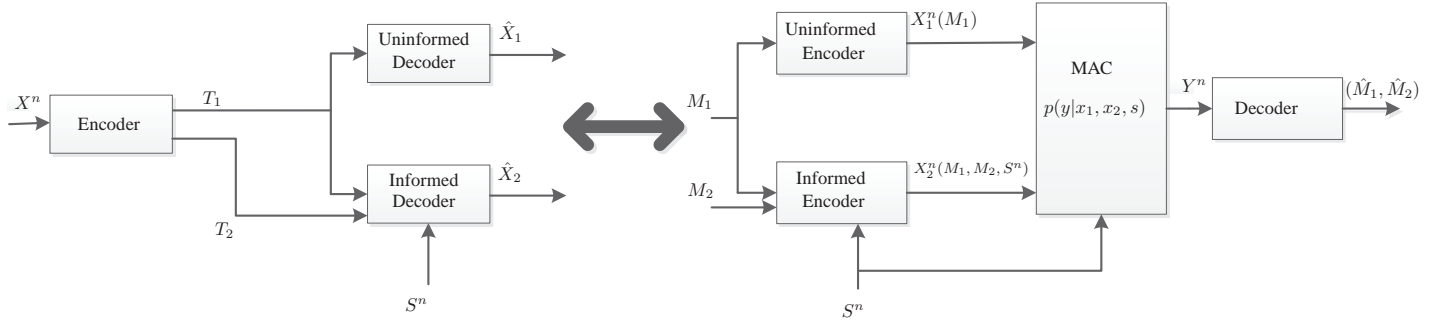}}
            \caption{The special case of the Stienberg Merhav rate distortion setting (on the left) compared to the  GGP MAC (on the right). Once again, taking out the channel block in the MAC, we have an exact mirror reflection between the settings.} \label{ShlomoVsKaspi}
            \psfragscanoff
        \end{psfrags}
     \end{center}
 \end{figure}

Using the same set of rules for the duality transformation principles obtained in Table \ref{DualityTable} we can establish a duality between the GGP MAC and the Stienberg-Merhav rate distortion setting. These results are a straightforward extension of the duality results of the Action-MAC and its dual rate distortion model, obtained by dismissing the action $A$.

\begin{center}
\begin{table}[h!]
\caption{Corner points $(R_1,R_2)$ of the GGP MAC setting and the Stienberg-Merhav rate distortion setting}\label{CornerPointsTable3}
\hspace{20mm}\begin{tabular}{|c||c c|}
\hline
& \textbf{$R_1$} & \textbf{$R_2$}\\ \hline\hline
\textbf{GGP} & $I(X_1;Y)+ I(Y;U|X_1) - I(S;U|X_1)$ & $0$\\
\textbf{MAC} & $I(X_1;Y)$ & $ I(U;Y|X_1) - I(U;S|X_1)$\\ \hline
\textbf{Stienberg-} & $I(X;\hat{X}_1)+I(X;U|A,\hat{X}_1)-I(S;U|\hat{X}_1)$ & $0$ \\
\textbf{Merhav} & $I(\hat{X}_1;X)$ & $I(U;X|\hat{X}_1) - I(U;S|\hat{X}_1)$\\
\hline
\end{tabular}
\end{table}
\end{center}

\subsubsection{Duality between the point-to-point action-dependent channel and source coding with side information ``Vending Machine''}
\begin{figure}[h!]
    \begin{center}
        \begin{psfrags}
            \psfragscanon
            \psfrag{A}[][][0.7]{\ \ \ \ \ \ \ \ \ \ \ \ \ \ \ \ Encoder}
            \psfrag{B}[][][0.7]{\ \ \ \ \ \ \ \ \ \ \ \ \  \ \ \  Decoder}
            \psfrag{C}[][][0.7]{\ \ \ \ \ \ \ \ \ \  Channel}
            \psfrag{D}[][][0.7]{\ \ \ \ \ \ \ \ \ \ \ \ $p(y|x,s)$}
            \psfrag{E}[][][0.7]{\ \ \ \ \ \ \ \ \ $p(s|a)$}
            \psfrag{F}[][][0.7]{\ \ \ \ \ \ \ \ \ \  $A(M)$}
            \psfrag{G}[][][0.7]{\ \ \ \ $S^n$}
            \psfrag{H}[][][0.7]{\ \ \ \ \ $X^n(M)$}
            \psfrag{I}[][][0.7]{\ \ \ \ \ \ $Y^n$}
            \psfrag{J}[][][0.7]{$M$}
            \psfrag{K}[][][0.7]{$\hat{M}$}
            \psfrag{L}[][][0.7]{\ \ \ \ \ $X^n$}
            \psfrag{M}[][][0.7]{\ \ \ \ \ \ \ \ \ $T(X^n)$}
            \psfrag{N}[][][0.7]{$\hat{X}$}
            \psfrag{O}[][][0.7]{\ \ \   $A(T)$}
            \psfrag{P}[][][0.7]{\ \ \  $S^n$}
            \psfrag{Q}[][][0.7]{\ \ \ \ \ \ \ \ \ \ \ Vender}
            \psfrag{R}[][][0.7]{\ \ \ \ \ \ \ \ \ \ \  $p(s|a,x)$}
            \centerline{\includegraphics[scale = .6]{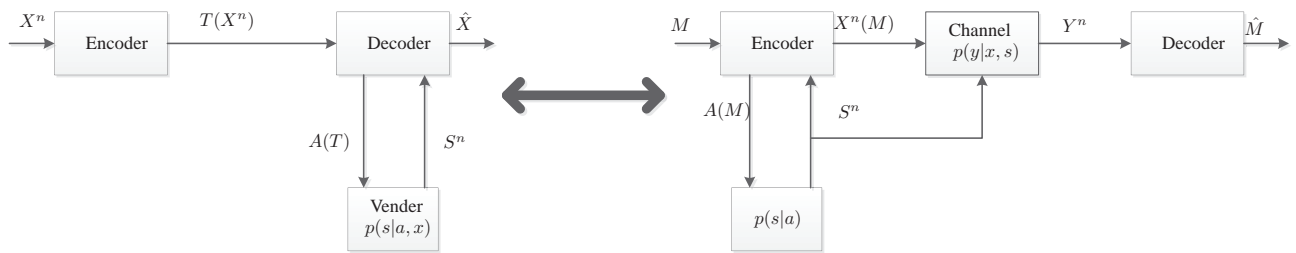}}
            \caption{The action-dependent point-to-point channel (on the right) compared to the source coding with side information ``vending machine'' at the decoder (on the left). Again, we see that if we take out the channel block from the point-to-point model, we have an exact mirror reflection between the two settings.}\label{VendingVsAcion}
            \psfragscanoff
        \end{psfrags}
     \end{center}
 \end{figure}

 \begin{center}
\begin{table}[h!]
\caption{Duality results of the action-dependent point-to-point channel vs. source coding with
side information ``vending machine''}\label{actionVSvending}
\begin{center}
\begin{tabular}{|c||c|}
\hline
\textbf{Action-dependent} & \textbf{Source coding with side information}\\
\textbf{point-to-point channel} & \textbf{``vending machine''}\\   \hline\hline
 Capacity expression & Rate distortion function\\
$C=I(A;Y)+I(U;Y|A)-I(U;S|A)$ & $R(D)=I(A;X)+I(U;X|A)-I(U;S|A)$\\ \hline \hline
\end{tabular}
\end{center}
\label{t_dual_point_to_point}
\end{table}
\end{center}

The duality between these two settings was noted first in \cite{KittichokechaiTobias}. Similarly
to the previous compression, the duality transformation principles presented in Table
\ref{DualityTable} are used to determine the duality between the action-dependent point-to-point
channel and the source coding with side information ``vending machine'' model. It is particularly
interesting to observe that the capacity expression and the rate distortion function are
equivalent with some renaming of variables, as presented in Table \ref{t_dual_point_to_point}.


\section{Conclusions}\label{Conclusion}
In multi-user communication systems today, there is an increasing demand for high data rates.
Therefore, it is essential to find the fundamental limits of channel models in order to benefit
from the channel structure. Motivated by this growing need, we extended the study of the MAC. We
consider a MAC with two encoders: an encoder which has access to a common message and an encoder
which has access to a common message as well as a private message. This second encoder can take
an action, dependent on both of the messages, that affects the formation of the states. We find a
single letter characterization of the capacity region for this channel. This is done by using a
three-stage binning coding scheme. We focus specifically on the Gaussian case, find an
characterize the capacity region and plot some computed results. In the process of this analysis,
we find the capacity expression for the Gaussian point-to-point action-dependent channel, which
was left open in \cite{Tsachy-Weissman} and is independently reported in \cite{Choudhuri-Mitra}.
In addition, we explore the dual rate distortion for our Action-MAC model. We give duality
principles that enable us to effectively estimate the result of one problem, given the result of
the other. Using the insight gained by this duality, we discover further duality comparisons
between previously solved settings in source coding and channel coding.

On a final note, we remark that finding a closed, single letter expression of the capacity region
leads us to believe that such solutions exist for various extensions of this model. One practical
extension is a MAC where the informed encoder has access to a noisy observation of the channel
state, $\tilde{S}^n$, instead of the noiseless state sequence $S^n$. Such a setting can model a
pubic relay, where the informed encoder observes a noisy version of the pubic relay's output.
Another possible extension is where the non informed encoder takes the action rather than the
informed encoder. This setting corresponds to the case that the non informed encoder influences
the public relay, while the public relay is observed (possibly with noise) by the informed
encoder.

\appendices                 
\section{Proof of Lemma \ref{lemma1}}\label{LemmaProof}
Let us prove that the capacity region described in {\it Theorem 1}, given in (\ref{CapacityRegion}), is convex.
\begin{proof}
Let $P_i,i=1,2,3$ be three distributions of the form (\ref{distribution}), i.e.:
\begin{align*}
P_i(a,s,u,x_1,x_2,y) = P_i(x_1)P_i(a|x_1)P_i(s|a)P_i(u|s,a,x_1)P_i(x_2|x_1,s,u)P_i(y|x_1,x_2,s).
\end{align*}
The given distributions induce the quantities:
\begin{equation}
\Big{(}I_i(U;Y|X_1), I_i(U;S|A,X_1), I_i(X_1,U;Y), I_i(X_1,U;S|A)\Big{)},
\end{equation}
for $i=1,2,3$, respectively. Let $\gamma$ be a number characterized by $0\leq\gamma\leq1$ and $\overline{\gamma}= 1-\gamma$, where $P_3=\gamma P_1 +\overline{\gamma}P_2$. In addition, let $Q$ be a binary random variable, where $P(q=1)=\gamma$ and $P(q=2)=\overline{\gamma}$, independent of $A,S,U,X_1,X_2$. The distribution of $A,S,U,X_1,X_2$ is according to $P_1$ when $q=1$ and according to $P_2$ when $q=2$. Furthermore, we denote $\tilde{U}=(U,Q)$. Note that by marginalizing $P(x_1)P(a|x_1)P(s|a)P(\tilde{u}|s,a,x_1)P(x_2|x_1,s,u)P(y|x_1,x_2,s)$ over $Q$, we obtain $P_3$ in the form of (\ref{distribution}). This is due to the fact that we can write $P(\tilde{u}|s,a,x_1)$ equivalently as $P(q)P(u|q,s,a,x_1)$. Finally, from the set of equalities:
\begin{eqnarray}
\gamma I_1(U;Y|X_1)+\overline{\gamma} I_2(U;Y|X_1)=I_3(U;Y|X_1,Q)&=&I_3(\tilde{U};Y|X_1),\nonumber\\
\gamma I_1(U;S|X_1,A)+\overline{\gamma} I_2(U;S|X_1,A)=I_3(U;S|X_1,A,Q)&=&I_3(\tilde{U};S|X_1,A),\nonumber\\
\gamma I_1(X_1,U;Y)+\overline{\gamma} I_2(X_1,U;Y)=I_3(X_1,U;Y,Q)&=&I_3(X_1,\tilde{U};Y|A),\nonumber\\
\gamma I_1(X_1,U;S|A)+\overline{\gamma} I_2(X_1,U;S|A)=I_3(X_1,U;S|A,Q)&=&I_3(X_1,\tilde{U};S|A),\nonumber
\end{eqnarray}
we derive that region is, indeed, convex.

\end{proof}


\section{Alternative Proof to the Gaussian Action-Dependent Point-to-Point Channel}\label{AppendixPTPGauss}

Weissman  \cite{Tsachy-Weissman} first introduced the Gaussian channel model for the action-dependent point-to-point channel illustrated in Fig \ref{PTPchannel}. He derived an achievable scheme  for a lower bound on the capacity of the channel. We would like to show that the expression (\ref{GaussainPTP22}) is the capacity of this setting and that the lower bound presented in \cite{Tsachy-Weissman} is an equivalent expression. In order to do so, we look at the Gaussian channel model presented in \cite{Baruch-Shamai-Verd}, namely the Gaussian Generalized Gel'fand-Pinsker (GGP) MAC. We examine the case where only a common message is sent over the channel. In this setting, we have an informed encoder which has access to the common message in addition to knowledge of the channel states non-causally. We also have an uninformed encoder in the sense that it does not know the channel states. This setting is illustrated in Fig \ref{GGPchannel}. The capacity region for this model was solved in \cite{Baruch-Shamai-Verd}. We use the result obtained in \cite{Baruch-Shamai-Verd} to prove the capacity of the action-dependent point-to-point channel. We note that similar results to this alternative proof were obtained simultaneously and independently by Choudhuri and Mitra in \cite{Choudhuri-Mitra}.

\begin{figure}[h!]
\begin{center}
    \psfrag{A}{$M$}
    \psfrag{B}{$S^n=A^n(M)+W^n$}
    \psfrag{C}{$X^n(M)$}
    \psfrag{D}{$Y^n=X^n+A^n+W^N+N^n$}
    \psfrag{E}{$\hat{M}$}
    \psfrag{F}{\scriptsize{Action Encoder}}
    \psfrag{G}{\scriptsize{Channel Encoder}}
    \psfrag{H}{\ \ \scriptsize{Gaussian channel}}
    \psfrag{I}{\ \ $p(s|a)$}
    \psfrag{J}{\ \ \scriptsize{Decoder}}
\includegraphics[scale = 0.45]{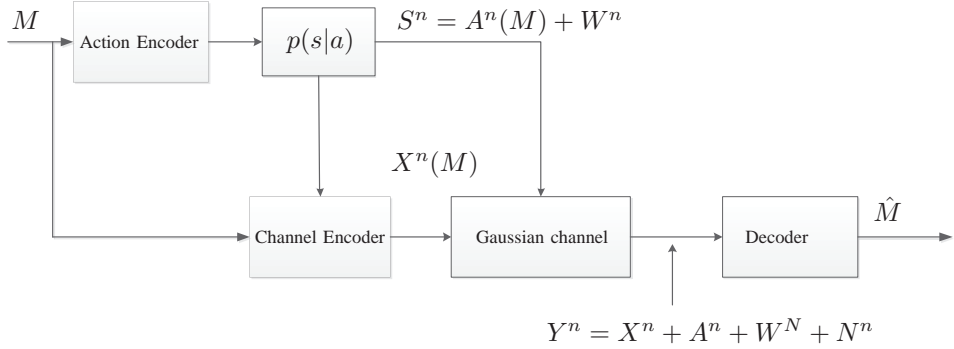}
\caption{Point-to-point channel with action-dependent states.} \label{PTPchannel}
\end{center}
\end{figure}

\begin{center}
\begin{figure}[h!]
        \psfrag{A}[][][0.7]{\ \ \ \ \ \ \ \ \ \ \ \ \ \ \ \  Uninformed}
        \psfrag{B}[][][0.7]{\ \ \ \ \ \ \ \ \ \ \ \ \ \ \ \  Informed}
        \psfrag{C}[][][0.7]{\ \ \ \ \ \ \ \ \ \ \ \ \ \ \  Gaussian}
        \psfrag{D}[][][0.7]{\ \ \ \ \ \ \ \ \ \ \ \ \ \ \  MAC}
        \psfrag{E}[][][0.7]{\ \ \ \ \ \ \ \ \ \ \  Decoder}
        \psfrag{F}[][][0.7]{\ \ \ \ \ \ {\footnotesize{$X_1^n(M)$}}}
        \psfrag{G}[][][0.7]{\ \ \ \ \ \ \  {\footnotesize{$X_2^n(M,W^n)$}}}
        \psfrag{H}[][][0.7]{\ \ \ \ \ \ \ \ \ \ \ \ \ \ \ \ \ \ \ \ \ \ \ \ \ \ \ \ \ \ \ \ \ $Y^n=X_1(M)^n+X_2^n(M,W^n)+W^n+Z^n$}
        \psfrag{I}[][][0.7]{\ \ \ \ \ \  $W^n$}
        \psfrag{J}[][][0.7]{$M$}
        \psfrag{L}[][][0.7]{\ \ \ $\hat{M}$}
        \psfrag{M}[][][0.7]{\ \ \ \ \ \ \ \ \ \ \ \ \ \ \ \ \ Encoder}
        \psfrag{N}[][][0.7]{\ \ \ \ \ \ \ \ \ \ \ \ \ \ \ \ \  Encoder}
        \centerline{\includegraphics[scale = 0.45]{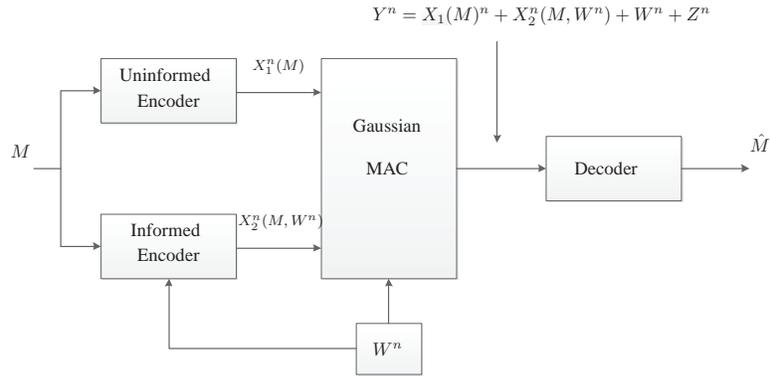}}
        \caption{MAC with state information at one encoder and a common message} \label{GGPchannel}
      \end{figure}
\end{center}

\subsection{Alternative Proof of the Capacity of the Action-Dependent Point-to-Point Channel}

The main idea of the proof is to show that we can obtain a one-to-one correspondence between the action-dependent point-to-point channel introduced in \cite{Tsachy-Weissman} and the MAC with state information known to one encoder channel, presented in \cite{Baruch-Shamai-Verd}.

\begin{proof}
\par We briefly recall the channel model presented in \cite{Baruch-Shamai-Verd}. The channel is given by:
\begin{equation}
Y_i=X_{1,i}+X_{2,i}+W_i+N_i.
\end{equation}
Here, the states are denoted as $W^n$, where the noise processes $W_i$ and $N_i$ are zero-mean Gaussian i.i.d. with $E[N_i^2]=N$ and $E[W_i^2]=Q$. Individual power constraints are considered:
\begin{equation}
\frac{1}{n}\sum_{i=1}^n X_{1,i}^2\leq P_1,\ \ \frac{1}{n}\sum_{i=1}^n X_{2,i}^2\leq P_2.
\end{equation}

Now, we would like to show that the setting of the Gaussian action-dependent point-to-point channel, illustrated in Fig \ref{PTPchannel} , is a special case of the Gaussian GGP channel with only a common message, illustrated in Fig. \ref{GGPchannel}. Let us take $X^n_1(M)$ as $A^n(M)$, and $X_2(M)$ as $X(M)$. We can look at the block of the ``Action Encoder'' as the ``Uninformed Encoder'' and the ``Channel Encoder'' as ``Informed Encoder''.
\begin{center}
\begin{tabular}{|c|c|}
\hline
Action-dependent p-t-p channel & GGP channel with common message\\ \hline
$A^n$ & $X_1^n$ \\
$X^n$ & $X_2^n$ \\
$f_A:\mathcal{M}\rightarrow \mathcal{A}^n$ & $f_{X_1}:\mathcal{M}\rightarrow \mathcal{X}_1^n$\\
$f_X:\mathcal{M}\times\mathcal{S}^n\rightarrow \mathcal{X}^n$ & $f_{X_2}:\mathcal{M}\times\mathcal{S}^n\rightarrow \mathcal{X}_2^n$\\ \hline
\end{tabular}
\end{center}
 \begin{eqnarray}
f_A:\mathcal{M}\rightarrow \mathcal{A}^n &\Rightarrow& f_{X_1}:\mathcal{M}\rightarrow \mathcal{X}_1^n\nonumber\\
f_X:\mathcal{M}\times\mathcal{S}^n\rightarrow \mathcal{X}^n &\Rightarrow& f_{X_2}:\mathcal{M}\times\mathcal{S}^n\rightarrow \mathcal{X}_2^n\nonumber
\end{eqnarray}
Notice that we do not lose any of the properties of the setting by looking at the action-dependent point-to-point channel as a Gaussian GGP MAC. Both the ``Channel Encoder'' and ``Informed Encoder'' have knowledge of $S^n$ and the channel inputs and the channel outputs are the same. The covariance matrix $\Sigma_{X_1,X_2,S,N}$ becomes:
\begin{equation}
\Sigma_{X_1,X_2,W,N}= \left(  \begin{array}{cccc}  P_1 & \sigma_{12} & 0 & 0 \\ \sigma_{12} & P_2 & \sigma_{2W} & 0 \\ 0 & \sigma_{2W} & Q & 0 \\ 0 & 0 & 0 & \sigma_N^2 \end{array}  \right)\ \Leftrightarrow \  \Sigma_{A,X,W,N}= \left(  \begin{array}{cccc}  P_A & \sigma_{XA} & 0 & 0 \\ \sigma_{XA} & P_X & \sigma_{XW} & 0 \\ 0 & \sigma_{XW} & Q & 0 \\ 0 & 0 & 0 & \sigma_N^2 \end{array}  \right).
\end{equation}

Conversely, we show that the Gaussian GGP channel, with only a common message, is a special case of the Gaussian action-dependent point-to-point channel. Let us take $A^n(M)$ as $X^n_1(M)$ and $X(M)$ as $X_2(M)$. So, we can look at the ``Uninformed Encoder'' block as the ``Action Encoder'' and the ``Informed Encoder'' block as the ``Channel Encoder''.
 \begin{eqnarray}
f_{X_1}:\mathcal{M}\rightarrow \mathcal{X}_1^n &\Rightarrow& f_A:\mathcal{M}\rightarrow \mathcal{A}^n\nonumber\\
f_{X_2}:\mathcal{M}\times\mathcal{S}^n\rightarrow \mathcal{X}_2^n &\Rightarrow& f_X:\mathcal{M}\times\mathcal{S}^n\rightarrow \mathcal{X}^n\nonumber
\end{eqnarray}

In conclusion, we obtained a one-to-one correspondence between the Gaussian action-dependent point-to-point channel and the Gaussian GGP channel with only a common message. The capacity region for the Gaussian GGP channel was found in \cite{Baruch-Shamai-Verd} and is equal to:
\begin{eqnarray}
&&R_2\leq \frac{1}{2}\log\Big{(}1+ \frac{P_2(1-\rho_{12}^2-\rho_{2S}^2)}{N}\Big{)}\nonumber\\
&&R_1+R_2\leq \frac{1}{2}\log\Big{(}1+ \frac{P_2(1-\rho_{12}^2-\rho_{2S}^2)}{N}\Big{)}+ \frac{1}{2}\log\Big{(}1+ \frac{(\sqrt{P_1}+\sqrt{P_2})^2}{P_2(1-\rho_{12}^2-\rho_{2S}^2)+(\sigma_W +\rho_{2S}\sqrt{P_2})^2N}\Big{)},\nonumber\\\label{GGPregion}
\end{eqnarray}
where
\begin{equation}
\rho_{12}=\frac{\sigma_{12}}{\sqrt{P_1P_2}},\ \ \rho_{2W}=\frac{\sigma_{2W}}{\sqrt{P_2Q}}.
\end{equation}
\begin{eqnarray}
\rho_{12}^2+\rho_{2W}^2\leq 1.
\end{eqnarray}
Hence, the capacity for the Gaussian action-dependent point-to-point channel can be achieved by substituting the following transformations in the Gaussian GGP capacity expression. Here, we have $M_2=0$, thus $R_2=0$, $P_1\rightarrow P_A$, $P_2\rightarrow P_X$, $\rho_{12}\rightarrow \rho_{XA}$ and $\rho_{2W}\rightarrow \rho_{XW}$. Substituting these expressions into (\ref{GGPregion}) we get
\begin{eqnarray}
R_1\leq \frac{1}{2}\log\Big{(}1+ \frac{P_X(1-\rho_{XA}^2-\rho_{XW}^2)}{N}\Big{)}+ \frac{1}{2}\log\Big{(}1+ \frac{(\sqrt{P_A}+\rho_{12}\sqrt{P_X})^2}{P_X(1-\rho_{XA}^2-\rho_{XW}^2)+(\sigma_W +\rho_{XW}\sqrt{P_X})^2+N}\Big{)}.
\end{eqnarray}
where
\begin{equation}
\rho_{XA}=\frac{\sigma_{XA}}{\sqrt{P_XP_A}},\ \ \rho_{XW}=\frac{\sigma_{XW}}{\sqrt{P_XQ}}.
\end{equation}
\begin{eqnarray}
\rho_{XA}^2+\rho_{XW}^2\leq 1.
\end{eqnarray}
\end{proof}


\subsection {Remarks}
 Recall the result for the achievable rate introduced in Weissman's paper \cite{Tsachy-Weissman} for the Gaussian point-to-point channel with action-dependent states
    \begin{equation}
    C_G=\frac{1}{2}\log\Big{(}\max_{(\alpha,\gamma):\alpha^2P_A+\gamma^2Q\leq P_X}\frac{[(1+2\alpha)P_A +P_X+1+(1+2\gamma)Q][1+(P_X-(\alpha^2P_A+\gamma^2Q))]}{P_X-\alpha^2P_A+1+Q(1+2\gamma)}\Big{)}.\label{LowerBound}
    \end{equation}
    We would like to show that our result, (\ref{GaussainPTP22}), is consistent with the result obtained in \cite{Tsachy-Weissman}. Here, $X$ is defined as:
    \begin{equation}
    X=\alpha A + \gamma W +G
    \end{equation}
    where $\alpha^2 P_A+\gamma^2Q\leq P_X$ and $G\sim N(0,P_X-(\alpha^2 P_A+\gamma^2Q))$. Therefore we have:
    \begin{eqnarray}
    \sigma_{XA}&=&\alpha P_A\nonumber\\
    \sigma_{XW}&=&\gamma Q\nonumber\\
    \rho_{XA}&=&\frac{\sigma_{XA}}{\sqrt{P_AP_X}}=\alpha\frac{\sqrt{P_A}}{\sqrt{P_X}}\nonumber\\ \rho_{XW}&=&\frac{\sigma_{XW}}{\sqrt{P_XQ}}=\gamma\frac{\sigma_W}{\sqrt{P_X}}.\label{expressions}
    \end{eqnarray}
    Notice that now the constraint
    \begin{eqnarray}
    \rho_{XA}^2+\rho_{XW}^2\leq 1\nonumber
    \end{eqnarray}
    becomes
    \begin{eqnarray}
    \alpha^2 P_A+\gamma^2Q\leq P_X.\nonumber
    \end{eqnarray}
    Now, substituting the expressions in (\ref{expressions}) into the capacity region (\ref{GaussainPTP22}), we achieve the lower bound (\ref{LowerBound}) found in \cite{Tsachy-Weissman} and have shown that it is, indeed, tight. Therefore, the capacity expression in (\ref{LowerBound}) is equivalent to the expression in (\ref{GaussainPTP22}).


\bibliographystyle{unsrt}
\bibliography{Ref-Mine}

\end{document}